\documentclass[a4paper]{IEEEtran}
\IEEEoverridecommandlockouts

\usepackage{amsmath}
\usepackage{amssymb}
\usepackage{amsthm}
\usepackage{mathtools}
\usepackage{xcolor}
\usepackage{cite}
\usepackage{thmtools}
\usepackage{thm-restate}
\usepackage{cancel}
\usepackage{tikz}
\usetikzlibrary{matrix,arrows,arrows.meta,positioning,calc,external,shapes}

\usepackage{arydshln}

\newif\ifcomment
\commentfalse %
\usepackage{xcolor}

\definecolor{TUMBlue}{RGB}{0,101,189} %
\definecolor{TUMBlueDark}{RGB}{0,82,147} %
\definecolor{TUMBlueLight}{RGB}{152,198,234} %
\definecolor{TUMBlueMedium}{RGB}{100,160,200} %

\definecolor{TUMIvory}{RGB}{218,215,203} %
\definecolor{TUMGreen}{RGB}{162,173,0} %
\definecolor{TUMGray}{gray}{0.6} %
\definecolor{TUMGrayDark}{gray}{0.3} %

\definecolor{TUMGreenDark}{RGB}{0,124,48} %
\definecolor{TUMRed}{RGB}{196,7,27} %
\definecolor{TUMOrange}{RGB}{227,114,34} %

\newtheorem{definition}{Definition}
\newtheorem{lemma}{Lemma}
\newtheorem{theorem}{Theorem}
\newtheorem{conjecture}{Conjecture}

\newtheorem{remark}{Remark}

\newcommand{\tpir}{\ensuremath{\mathrm{TPIR}}}
\newcommand{\spir}{\ensuremath{\mathrm{SPIR}}}
\newcommand{\tspir}{\ensuremath{\mathrm{TSPIR}}}
\newcommand{\tbspir}{\ensuremath{\mathrm{TBSPIR}}}
\newcommand{\tbpir}{\ensuremath{\mathrm{TBPIR}}}
\newcommand{\mds}{\ensuremath{\mathrm{MDS}}}
\newcommand{\cT}{\ensuremath{\mathcal{T}}}
\newcommand{\cN}{\ensuremath{\mathcal{N}}}
\newcommand{\cI}{\ensuremath{\mathcal{I}}}

\newcommand{\cQ}{\ensuremath{\mathcal{Q}}}

\newcommand{\F}{\ensuremath{\mathbb{F}}}
\newcommand{\cH}{\ensuremath{\mathcal{H}}}
\newcommand{\cR}{\ensuremath{\mathcal{R}}}
\newcommand{\cC}{\ensuremath{\mathcal{C}}}
\newcommand{\cD}{\ensuremath{\mathcal{D}}}
\newcommand{\cS}{\ensuremath{\mathcal{S}}}
\newcommand{\cF}{\ensuremath{\mathcal{F}}}
\newcommand{\cB}{\ensuremath{\mathcal{B}}}
\newcommand{\bx}{\ensuremath{\mathbf{x}}}
\newcommand{\sfa}{\ensuremath{\mathsf{a}}}
\newcommand{\sfb}{\ensuremath{\mathsf{b}}}
\newcommand{\sfc}{\ensuremath{\mathsf{c}}}

\newcommand{\sfD}{\ensuremath{\mathsf{D}}}
\newcommand{\bX}{\ensuremath{\mathbf{X}}}
\newcommand{\bY}{\ensuremath{\mathbf{Y}}}
\newcommand{\bG}{\ensuremath{\mathbf{G}}}
\newcommand{\bQ}{\ensuremath{\mathbf{Q}}}
\newcommand{\bA}{\ensuremath{\mathbf{A}}}
\newcommand{\bB}{\ensuremath{\mathbf{B}}}
\newcommand{\bC}{\ensuremath{\mathbf{C}}}
\newcommand{\bD}{\ensuremath{\mathbf{D}}}
\newcommand{\bE}{\ensuremath{\mathbf{E}}}
\newcommand{\bI}{\ensuremath{\mathbf{I}}}
\newcommand{\bP}{\ensuremath{\mathbf{P}}}
\newcommand{\bW}{\ensuremath{\mathbf{W}}}

\newcommand{\bS}{\ensuremath{\mathbf{S}}}
\newcommand{\bq}{\ensuremath{\mathbf{q}}}
\newcommand{\ba}{\ensuremath{\mathbf{a}}}
\newcommand{\bb}{\ensuremath{\mathbf{b}}}
\newcommand{\bc}{\ensuremath{\mathbf{c}}}
\newcommand{\bd}{\ensuremath{\mathbf{d}}}
\newcommand{\be}{\ensuremath{\mathbf{e}}}
\newcommand{\bz}{\ensuremath{\mathbf{z}}}

\newcommand{\ceil}[1]{\ensuremath{\left\lceil #1 \right\rceil}}
\newcommand{\floor}[1]{\ensuremath{\left\lfloor #1 \right\rfloor}}

\DeclareMathOperator{\rank}{rank}
\DeclareMathOperator{\supp}{supp}
\DeclareMathOperator{\colsupp}{colsupp}

\newcommand{\myspan}[1]{\left\langle #1 \right\rangle}
\newcommand{\myspanCol}[1]{\left\langle #1 \right\rangle_{\mathsf{col}}}
\newcommand{\myspanRow}[1]{\left\langle #1 \right\rangle_{\mathsf{row}}}

\newcommand{\newName}{full support-rank}
\newcommand{\NewName}{Full support-rank}

\begin{document}

\title{Towards the Capacity of Private Information Retrieval from Coded and Colluding Servers}

\author{%
    \IEEEauthorblockN{Lukas Holzbaur, Ragnar Freij-Hollanti, Jie  Li, Camilla Hollanti}
    \thanks{The results related to symmetric PIR and strongly linear PIR were presented at the 2019 IEEE Information Theory Workshop \cite{Holzbaur2019}. The results concerning \newName{} PIR capacity are new, and more proofs are added with respect to \cite{Holzbaur2019}.

    The work of L. Holzbaur was supported by the Technical University of Munich -- Institute for Advanced Study, funded by the German Excellence Initiative and European Union 7th Framework Programme under Grant Agreement No. 291763 and the German Research Foundation (Deutsche Forschungsgemeinschaft, DFG) under Grant No. WA3907/1-1.
    The work of C. Hollanti was supported by the Academy of Finland, under Grants No. 336005 and 318937, and by the Technical University of Munich -- Institute for Advanced Study, funded by the German Excellence Initiative and the EU 7th Framework Programme under Grant Agreement No. 291763, via a Hans Fischer Fellowship.

    L.~Holzbaur is with the Institute for Communications Engineering, Technical University of Munich, Germany. Ragnar Freij-Hollanti and Camilla Hollanti are with the Department of Mathematics and Systems Analysis, Aalto University, Finland. Jie Li was with the Department of Mathematics and Systems Analysis, Aalto University, Finland.

    Emails: lukas.holzbaur@tum.de, \{ragnar.freij, camilla.hollanti\}@aalto.fi, jieli873@gmail.com}
 }

\maketitle

\begin{abstract}

In this work,  two practical concepts related to private information retrieval (PIR) are introduced and coined \emph{\newName{}} PIR and \emph{strongly linear} PIR. Being of \newName{} is a technical, yet natural condition required to prove a converse result for a capacity expression and satisfied by almost all currently known  capacity-achieving schemes, while
strong linearity is a practical requirement enabling implementation over small finite fields with low subpacketization degree.

Then, the  capacity of MDS-coded, linear, \newName{} PIR in the presence of colluding servers is derived, as well as the capacity of symmetric, linear PIR with colluding, adversarial, and nonresponsive servers for the recently introduced concept of matched randomness. This positively settles the capacity conjectures stated by Freij-Hollanti~\emph{et al.} and Tajeddine~\emph{et al.} in the presented cases.
It is also shown that, further restricting to strongly-linear PIR schemes with deterministic linear interference cancellation, the so-called star product scheme proposed by Freij-Hollanti \emph{et al.} is essentially optimal and induces no capacity loss.%

  \end{abstract}

\IEEEpeerreviewmaketitle

\section{Introduction} \label{sec:introduction}
User privacy has increased its importance together with the increasing usage of distributed services such as cloud storage and various peer-to-peer  networks. Recently, private information retrieval (PIR) \cite{chor1995private} in the context of coded storage has gained a lot of interest. With PIR, a user is able to download a desired file from a database or distributed storage system without revealing the identity of the file to the servers. Several PIR capacity results have been derived in various scenarios, \emph{e.g.}, for replicated storage \cite{Sun2017} and maximum distance separable (MDS) coded storage \cite{Banawan2018}, colluding servers \cite{Sun2016}, single-server PIR with side information \cite{Heidarzadeh2018single, Heidarzadeh2018singlemulti}, and symmetric PIR (SPIR)~\cite{sun2018sym,Wang2017adv,Wang2017col,Wang2016,Wang2019symmetric}.
\emph{Symmetric} refers to the property that the user is only able to decode the file that she has requested, and learns nothing about the other files.  We will denote nonsymmetric/symmetric PIR with $t$-collusion by  TPIR/TSPIR and with additional $b$ Byzantine (and possibly $r$ nonreponsive) servers by TBPIR/TBSPIR, respectively. It has also been shown that the MDS property is not necessary for achieving the MDS--PIR capacity \cite{Freij-Hollanti2019transitive, Kumar2019arblin}.

In this paper, we derive new results on the capacity for different PIR models. First, we will prove Conjecture~1 in \cite{Freij-Hollanti2017} for MDS-coded, linear, \newName{} PIR with colluding servers. After this, we will develop the concept of a strongly-linear PIR scheme, and prove the capacity of strongly-linear (nonsymmetric) PIR schemes for any number of files $m$. This also yields a proof in this practical special case for the conjecture stated in the asymptotic regime ($m\rightarrow\infty$) in  \cite[Conj.~1]{Tajeddine2018}. Finally, we prove Conjecture 2 in \cite{Tajeddine2018} for linear, symmetric PIR with coded, colluding, and adversarial servers for the case of \emph{matched} randomness as introduced in \cite{Wang2019symmetric} (see Section~\ref{sec:symmetricCap} for a rigorous definition).
We restate the conjectures later in this section for the ease of reading and numbering.

The main contribution of this paper is the proof for the capacity of MDS-coded, linear, \newName{} PIR with colluding servers. Nonrigorously, linearity refers to the property that the responses are obtained as a linear combination of the (encoded pieces of the) files stored at each node, with the scalar coefficients given by the entries of the received query. While this restricts the PIR scheme in its generality, it appears to be a natural assumption to make, as to the best of out knowledge \emph{all} (asymptotically) capacity-achieving schemes fulfill this property \cite{Sun2017,Freij-Hollanti2017,Sun2018conj,Sun2016,Banawan2018,Tajeddine2018,Freij-Hollanti2019transitive,Kumar2019arblin,Banawan2019byzcoll,oliveira2019one}. The converse (upper bound) is given by Theorem~\ref{thm:coded-colluded}, and a scheme achieving this bound is given by applying the refinement and lift operation of \cite[Cor.~1]{oliveira2019one} to, \emph{e.g.}, the star product scheme \cite{Freij-Hollanti2017}. While the seemingly technical assumption of \newName{} (cf. Def.~\ref{def:newPIRproperty}) is unnecessary from the point of view of proving a general capacity result, we demonstrate its practical relevance in two important regards. Firstly, all capacity achieving schemes for the special cases of $k=1$ (uncoded storage) or $t=1$ (no collusion) given in \cite{Sun2017,sun2018sym,Sun2016,Banawan2018,Banawan2019byzcoll} fulfill this definition. Second, the only scheme for general parameters, introduced in \cite{oliveira2019one}, achieving this newly proved capacity is also of \newName{}\footnote{We note that the necessary assumption was \emph{not} made in the original paper \cite{oliveira2019one}, however, as we show in Appendix~\ref{app:liftedFix}, it is in fact required to hold for the scheme to be private.}.

Further, and maybe most importantly, the result provides insights towards what is required for proving a general capacity expression. To better illustrate this, we take a high-level look at existing schemes: In the ``simplest'' approach, as utilized in \cite{chor1995private,tajeddine2018private, Tajeddine2018, Freij-Hollanti2017}, privacy is achieved through ensuring that each $t$-tuple of servers receives a set of vectors uniformly distributed over the respective vector space. The advantage of these schemes is that they achieve the respective asymptotic PIR capacity (at least for the cases where it is known), are relatively simple, and allow for small subpacketization (see also Section~\ref{sec:strongly}). However, they fall short in achieving the capacity for a finite number of files.
The schemes able to achieve these capacities are based on querying for specific, carefully chosen pieces of (encoded) files. In this case, the queries received by $t$-tuples are no longer uniformly distributed over \emph{all} vectors since, for example, the all-zero vector will never be a query in this case. Similarly, the only general scheme achieving the new capacity for the coded-colluding case $k,t>1$, given in \cite{oliveira2019one}, is also based on constructing queries supported only on the positions corresponding to specific, carefully chosen files. Further, as shown in Appendix~\ref{app:liftedFix}, the natural choice to achieve privacy here, is requiring supported positions to be linearly independent. Our definition of \newName{} PIR (see Definition~\ref{def:newPIRproperty}) captures this linear independency of the queries that these schemes have in common. Thereby, the results we prove in the following show that in order to exceed the rate achieved by the scheme in  \cite{oliveira2019one}, it is \emph{necessary} for some restrictions of the queries to subsets of $t$ servers to be linearly dependent. To further support this argument, we show in Appendix~\ref{app:counterExample} that it is exactly this property that allows the scheme of \cite{Sun2018conj}, which is \emph{not} of \newName{}, to exceed the (thereby disproved in full generality) conjectured capacity of \cite[Conjecture~1]{Freij-Hollanti2017}.

Finally, the used transformation of the problem of linear PIR to the properties of the Khatri-Rao product of matrices results in a new formulation of the problem that might be useful for determining the general capacity of linear PIR, as discussed in Remark~\ref{rem:newPIRformulation}.

Nonrigorously, the rate of a PIR scheme with $m$ files is denoted and defined as
\begin{equation*}
R_m=\frac{\textrm{size of the desired file}}{\textrm{size of the total download}}\ .
\end{equation*}
We denote by  $C_m$  the \emph{capacity}, \emph{i.e.}, the largest achievable rate of a PIR scheme for $m$ files under some given constraints. A collection of schemes defined for a varying number of files has is said to have asymptotic rate
\begin{align*}
R\coloneqq\lim_{m\rightarrow\infty}R_m \ ,
\end{align*}
and is called asymptotically capacity achieving if
\begin{align*}
R=\lim_{m\rightarrow\infty}C_m \ .
\end{align*}

In Table~\ref{tab:capacity}, we summarize the known asymptotic capacity results relevant to this paper, as well as show the conjectured results \cite{Freij-Hollanti2017,Tajeddine2018} in red. We give  a precise problem setup as well as more rigorous definitions for the rate and capacity later in Section~\ref{sec:SystemSetup}.

\begin{table}\label{tab:capacity}
\begin{center}
\caption{Asymptotic capacity results and conjectures (in red). The maximum number of colluding / Byzantine / nonresponsive servers is denoted by $t,b,r$ respectively.}
\renewcommand{\arraystretch}{2}
\begin{tabular}{|l|cc|}
\hline
PIR model & $(n,k)$ MDS-coded PIR & Ref.  \\ \hline
& \textcolor{red}{$1-\frac{k+t+2b+r-1}{n}$} &\cite{Tajeddine2018}   \\
$b,r=0$ & \textcolor{red}{$1-\frac{k+t-1}{n}$} &\cite{Freij-Hollanti2017}\\
$k=1,r=0$ & $1-\frac{t+2b}{n}$ &\cite{Banawan2019byzcoll}\\
$t=1,b=r=0$ & $1-\frac{k}{n}$ &\cite{Banawan2018}  \\
$k=1,b=r=0$ & $1-\frac{t}{n}$ &\cite{Sun2016}  \\ \hline
\end{tabular}
\end{center}
\end{table}

\subsection{Notation}

Throughout the paper, we denote a finite field of $q$ elements by $\F_q$ or shortly $\F$.  For integers $a,b$ we write $[a,b]$ for the set of integers $\{i : a \leq i \leq b\}$ and if $a=1$ we neglect it, \emph{i.e.}, write $[1,b]=[b]$.
A code over $\F_q$ mapping $k$ information symbols to $n$ encoded symbols with minimum distance $d$ is denoted by $(n,k,d)$. Here, the length $n$ can also be thought of as the number of servers in the storage system.  Maximum distance separable (MDS) codes satisfying the Singleton bound with equality, \emph{i.e.}, $d=n-k+1$, are denoted by $(n,k)$. Linear codes are respectively denoted by $[n,k,d]$ and $[n,k]$, where the distinction from the set of integers $[a,b]$ will be clear from context.

In the following we will define several random variables that represent matrices in the setting of linear PIR. When treating these random variables, we use capital letters $W=\{W_{1}, W_{2},\ldots ,W_{n}\}$ and write $\supp(W)$ to denote the set of realizations of $W$ with nonzero probability. For integers $j,l \in [n]$ with $j\leq l$ denote $W_{j,...,l} = \{W_j,W_{j+1},\ldots,W_{l}\}$ and for a set of integers $\mathcal{T} \subseteq [n]$ denote $W_{\mathcal{T}} = \{W_j : j \in \mathcal{T}\}$. To establish the required technical results we will also need to treat these random variables, which then correspond to matrices, as a matrix of random variables, where each $W_j$ corresponds to a set of rows or columns. We denote such a matrix by $\bW = [\bW_1^\top, \bW_2^\top,\ldots ]^\top$ or $\bW = [\bW_1, \bW_2,\ldots ]$, respectively, where $\bW^\top$ denotes the transpose of $\bW$. In this matrix, the $j^{\mathrm{th}}$ block of rows/columns corresponds to $W_j$. To denote the restriction to the rows/columns of this matrix indexed by a set of integers $\cI$ we write $\bW[\cI,:]$, and similarly, $\bW[:,\cI]$ to denote the restriction to the respective columns.

Semantically, the rows/columns of such a matrix $\bW$ corresponding to each $W_j$ belong together. However, to avoid double indexing, we restrict ourselves to only use one method of indexing, \emph{i.e.}, either super-/subscripts or square brackets, at a time. When necessary, we refer to such sets of rows/columns as \emph{thick rows/columns} and to index them, we define a map from the indices of such thick rows/columns, to sets of normal rows/columns. For a set $\mathcal{I} \subseteq [n]$ define
\begin{align}
  \psi_{\beta}(\cI) = \bigcup\limits_{i \in \cI} \{(i-1)\beta +1,\ldots, i\beta\} \ . \label{eq:mapThickColumns}
\end{align}
Then, for an $m \times n \beta$ matrix $\bW = [\bW_1, \bW_2,\ldots, \bW_n ]$, where each $\bW_j$ is a $m \times \beta$ matrix, the restriction $\bW[:,\psi_{\beta}(\cI)]$ indexes the $|\cI|$ \emph{$\beta$-thick columns} given by $\cI$, where in this case a thick column is a submatrix consisting of $\beta$ consecutive columns of $\bW$. Note that this is equivalent to the set of random variables~$W_\cI$. The same notation is used to index thick rows. We use $\colsupp(\bW)$ to denote the set of indices of nonzero columns of $\bW$.
For the row and column span of a matrix $\bW$ we write $\myspanRow{\bW}$ and $\myspanCol{\bW}$, respectively.

For the reader's convenience, the notation used in this paper is summarized in the Table \ref{tab:notation} given in Appendix~\ref{app:notation}.

\subsection{Conjectures and contributions}

Let us now assume $n>k+t+2b+r-1$, where  $t,b,r$ refer to the number of colluding, Byzantine, and nonresponsive  servers, respectively.
The following conjectures describe a capture an observation that can be made for the cases where both are known, \emph{i.e.}, the cases where either $k=1$ or $t=1$. There, it can be seen that the symmetric capacity coincides with the asymptotic (in the number of files), nonsymmetric capacity.

\begin{conjecture}[\cite{Tajeddine2018}, Conjecture 1] The asymptotic capacity (as $m \rightarrow \infty$) of PIR from an $(n,k)$ MDS storage code with $t$-collusion, $b$ Byzantine servers, and $r$ nonresponsive servers is
  \begin{align*}
1-\frac{k+t+2b+r-1}{n} \ .
  \end{align*}
\end{conjecture}

\begin{conjecture}[\cite{Tajeddine2018}, Conjecture 2]\label{con:symCap}
 The capacity of SPIR from an $(n,k)$ MDS storage code with $t$-collusion, $b$ Byzantine servers, and $r$ nonresponsive servers is
 \begin{align*}
1-\frac{k+t+2b+r-1}{n} \ .
 \end{align*}
\end{conjecture}

\begin{remark}
In the original version of the above conjectures, the denominator is $n-r$ instead of $n$. This is due to assuming that we do not download anything from the nonresponsive serves (\emph{e.g.}, the request is dropped after a certain waiting time). Here, we also count the nonresponsive servers in the download cost, but point out that the results apply to both points of view.
\end{remark}
For the case of finite number of files, the observation of the capacity expressions for the known cases of either $k=1$ or $t=1$ naturally leads to the following conjecture.
\begin{conjecture}[\cite{Freij-Hollanti2017}, Conjecture 1]\label{tconj}
  Let $\cC$ be an $[n,k,d]$ code with a generator matrix $\bG$ that stores $m$ files via the distributed storage system $\bY=\bX\cdot\bG$, and fix $1\leq t\leq n-k$.  Any PIR scheme for $\bY$ that protects against any $t$ colluding servers has rate at most $R_m$,
  \begin{align*}
  R_m\leq\frac{1-\frac{k+t-1}{n}}{1-(\frac{k+t-1}{n})^m}\, \stackrel{m\rightarrow \infty}{\xrightarrow{\hspace*{1cm}}}\, 1-\frac{k+t-1}{n} \ .
  \end{align*}
\end{conjecture}

 However, Conjecture \ref{tconj} in its full extent was disproved in \cite{Sun2018conj}, where the authors exhibited an explicit PIR scheme for $m=2$ files distributed over $n= 4$ servers using a rate $1/2$ storage code, which protects against $t= 2$ collusion. This scheme has rate $3/5$, while the conjectured capacity was $4/7$. The proposed query scheme is not of \newName{}, see Appendix~\ref{app:counterExample} and does therefore not violate the results in this paper.

We refine Conjecture \ref{tconj} here by adding a technical requirement,
 and state this modified version as Theorem~\ref{thm:coded-colluded} for MDS-coded, linear, \newName{} PIR. The proof can be found in Sec.~\ref{subsec:converse}. We will show later in Sec. \ref{sec:strongly}, Theorem~\ref{thm:strongly}, that the asymptotic capacity expression holds for any strongly-linear PIR scheme regardless of the number of files, under the assumption that all servers respond and their responses have the same size. See Sec.~\ref{sec:strongly} for more details.
\begin{restatable}{mythm}{mainthm}\label{thm:coded-colluded}
Let $n,k,t$, and $m$ be integers with $n>k+t-1$ and $m\geq 2$. The capacity of linear, \newName{} PIR from $[n,k]$ MDS-coded storage with $t$ colluding servers, where all servers are honest and responsive, is given by
  \begin{equation*}
    C_{\tpir}^{[n,k]-\mds} = \frac{1-\frac{k+t-1}{n}}{1-\left(\frac{k+t-1}{n}\right)^m}\, \stackrel{m\rightarrow \infty}{\xrightarrow{\hspace*{1cm}}}\, 1-\frac{k+t-1}{n} \ .
  \end{equation*}
\end{restatable}

In what follows, we will prove Conjecture~2 for linear SPIR and Conjecture~\ref{tconj} for linear, MDS-coded, \newName{} PIR (that is, we prove Thm.~\ref{thm:coded-colluded}). We subsequently provide a proof for Conjecture~1 in the case that, in addition to the responses, also the query scheme and the (deterministic) interference cancellation are linear. We coin such a system \emph{strongly-linear}.

\begin{remark}\label{rem:practicality}

We would like to emphasize that strongly-linear schemes form a very relevant and practical case, namely the respective capacity result is known to be achievable \cite{Freij-Hollanti2017, Tajeddine2018} by a small field size $q\geq n$, which is that of a generalized Reed--Solomon code. Moreover, the subpacketization level is independent of $m$ and is (at most) quadratic in $n$ \cite[Eq.~(17)]{Freij-Hollanti2017}. This is in contrast to the schemes in~\cite{Sun2016, Sun2017,Banawan2018, zhang2018optimal}, where each file is assumed to be subdivided into a number of packets that grows exponentially with the number of files $m$.  It was shown in~\cite{zhang2018optimal} that an exponential (in $m$) number  of packets per file was necessary for a PIR scheme with optimal download rate, under the assumption that all servers respond to the queries and the responses have the same size. In \cite{zhou2020capacity} a scheme was presented that achieves the capacity with only $O(n)$ packets by making a weaker assumption on the size of the responses than in \cite{zhang2018optimal}.
\end{remark}

In Section~\ref{sec:symmetricCap} we prove the capacity of MDS-coded, symmetric PIR with colluding, byzantine, and nonreponsive servers, as stated in Conjecture~\ref{con:symCap}, for specific distributions of the randomness shared by the servers (for rigorous definitions and known results see Section~\ref{sec:symmetricCap}).

Finally, in Section~\ref{sec:subpacketization}, we show that, when assuming that files and reponses are over the same field, the rate achievable by strongly linear schemes is in fact optimal in some parameter regimes by generalizing the results of \cite{sun2017optimal}.

\section{Problem Setup and Known Results} \label{sec:SystemSetup}

We consider a distributed storage system with $n$ servers storing $m$ files $X=\{X^1,X^2,\ldots,X^m\}$, where each $X^l$ is a random variable over $\F^{\alpha \times k}$. Interpreted as a matrix, the \emph{data matrix} is denoted by $\bX$, where each block of $\alpha$ consecutive rows corresponds to a file. This matrix is encoded with an $(n,k)$ MDS storage code and server $j$ stores the $j^{\rm th}$ thick column of
\begin{equation*}
  \bY = \bX\cdot \bG =
  \begin{pmatrix}
    \bX^1\\ \bX^2 \\ \vdots \\ \bX^m
  \end{pmatrix} \cdot \bG =
  \begin{pmatrix}
    \bY_1^1 & \bY_2^1 & \dotsi & \bY_n^1 \\
    \bY_1^2 & \bY_2^2 & \dotsi & \bY_n^2 \\
    \vdots & \vdots & \ddots & \vdots \\
    \bY_1^m & \bY_2^m & \dotsi & \bY_n^m \\
  \end{pmatrix} \quad \in \F^{\alpha m \times n}\ ,
\end{equation*}
where $\bG$ is a generator matrix of the storage code. Note that $\bY^l=\bX^l \bG$ is the encoded version of the $l^{\rm th}$ file.
For $l\in [m]$, $j \in [n]$ and $\mathcal{I}\subset [n]$, we define $Y^l$, $Y_j$ and $Y_{\mathcal{I}}$ to be the random variables corresponding to the encoded version of $\bX^l$,\emph{ i.e.}, $\bY[\psi_{\alpha}(l),:]$, the $j^{\rm th}$ column of $\bY$, \emph{i.e.,} $\bY[:,j]$, and the restriction of $\bY$ to the columns indexed by $\mathcal{I}$, \emph{i.e.}, $\bY[:,\mathcal{I}]$, respectively.
We think of $\alpha$ as the number of {\em stripes} of each file, and each stripe is encoded independently of other stripes.
The $m$ files are independent and each consists of $k$ i.i.d. randomly drawn symbols from $\F_q^\alpha$, hence, for the entropies it holds that
\begin{align*}
  &H(X^i) = k\alpha \log q, \ \forall \ i \in [m]\\
  &H(X^1,\ldots,X^m) = mk\alpha \log q \ .
\end{align*}
We consider MDS codes, so every $k$ servers exactly recover the file, \emph{i.e.}, for any set $\mathcal{W} \subset [n]$ with $|\mathcal{W}| = k$ it holds that
\begin{align*}
  H(Y_{\mathcal{W}}) = H(X^1,\ldots,X^m) &= mk\alpha \log q\\
  H(X^1,\ldots,X^m|Y_{\mathcal{W}}) &= 0 \ .
\end{align*}
We also assume that the servers have access to a shared source of randomness, which has been shown \cite{gertner2000protecting} (see also \cite[Footnote~2]{sun2018sym}) to be required for enforcing the property of symmetry, \emph{i.e.}, ensuring that the user learns nothing about the files other than the requested file. Formally, let $\cS$ be a vector space over $\F$, and let
\begin{align*}
S=(S_1,\dots ,S_n)\in \cS^n
\end{align*}
be a random vector, where the symbols of $S_j$ may be used by the $j^{\rm{th}}$ server.

In a general PIR scheme, a user desiring the file with index $i$ picks the corresponding query
\begin{align*}
Q^i =  \left(Q_1^i, \ldots, Q_n^i\right)
\end{align*}
from the set of all possible queries $\cQ$, and sends $Q_j^i$ to the $j^{\rm th}$ server. Every server returns a response $A_j^i$ that is an $\beta$-tuple of symbols in $\F$. For a nonadversarial server, this response depends on the query $Q_j^i$, the symbols $Y_j$ stored at server $j$, and the randomness~$S$ shared by the servers, in a way known to the user. The list of responses from all servers for a given query is denoted by
\begin{align*}
A^i = \left(A_1^i,\ldots,A_n^i\right) \ .
\end{align*}
The desired file $X^i$ should now be recoverable from the responses, meaning that
\begin{equation}
\label{eq:correctness}
H(X^i | Q^i, A^i)=0 \ .
\end{equation}

In this work we only consider PIR schemes in which the query functions are linear.
\begin{definition}[Linear PIR]\label{def:linearPIR}
A PIR scheme is said to be \emph{linear} if
\begin{itemize}
    \item the query $Q^i$ can be represented as a matrix $\bQ^i \in F^{\alpha m \times \beta n}$, where each $\beta$-thick column $\bQ^i[:,\psi_\beta(j)]$ corresponds to the query $Q_j^i$ to server $j \in [n]$, and
\item the responses $A_j^i$ of server $j\in [n]$ are given by the vector
\begin{align*}\label{eq:linearPIR}
  \bA^i[:,\psi_\beta(j)] &= \big(\myspan{\bY[:,j],\bQ^i[:,(j-1)\beta+s]} \\
   &\hspace{2.2cm}+ \bS[:,(j-1)\beta+s]\big)_{s\in [\beta]} \\
  &= \bY[:,j]^\top \cdot \bQ^i[:,\psi_\beta(j)] + \bS[:,\psi_\beta(j)]\ ,
\end{align*}
where the vector $\bS\in\F^{1\times \beta n}$ depends on the randomness $S$ shared by the servers.
\end{itemize}
\end{definition}
Briefly and nonrigorously, in a linear PIR scheme each server receives $\beta$ query vectors and responds with the $\beta$ inner products between these vectors and the column of $\bY$ that it stores (possibly plus an additional symbol given by the shared randomness).
In the case of nonsymmetric PIR, the servers do not need any shared randomness and we may assume that~$\bS=\mathbf{0}$.

It is customary to think of the $\beta$ coordinates of the queries $Q_j^i$ as {\em iterations}. In this terminology, a linear PIR scheme consists of $\beta$ iterations, where in iteration $s$ the user sends for each $j\in [n]$ the query vector
\begin{align*}
\bQ^i[:,(j-1)\beta+s]\quad \in \F^{\alpha m \times 1}
\end{align*}
and receives a response row vector
\begin{align*}
  \big(\bA[:,(j-1)\beta+s]\big)_{j \in [n]} \quad \in \F^{1\times n}.
\end{align*}
It is easy to see that it is suboptimal to send linearly dependent queries to servers. However, in general, submatrices of the query matrix may indeed be nontrivially linearly dependent (see \cite{Sun2018conj} and Appendix~\ref{app:counterExample}), \emph{i.e.}, have supported columns that are linearly dependent. The technical assumption we make in the following, given below in Definition~\ref{def:newPIRproperty}, restricts all supported columns of the query for a subset of less than or equal to $t$ servers to be linearly independent, even when restricting to an arbitrary subset of files. We therefore coin these schemes as \newName{} PIR schemes.
\begin{definition}[\NewName{} PIR] \label{def:newPIRproperty}
  A linear PIR scheme is said to be of \emph{\newName{}} if for every query realization $\bQ^i=\bq \in \F^{\alpha m\times \beta n}$,
  any subset $\cT\subseteq[n]$ of $|\cT|\leq t$ servers, and any file index $j\in[m]$ it holds that
  \begin{align*}
        \rank(\bq[\psi_\alpha(j), \psi_\beta(\cT)])=|\colsupp(\bq[\psi_\alpha(j), \psi_\beta(\cT)])|.
  \end{align*}
\end{definition}

Most PIR schemes in the literature are indeed of \newName{}, including those in~\cite{Sun2017,sun2018sym,Sun2016,Banawan2018,Banawan2019byzcoll, oliveira2019one,zhou2020capacity,li2020towards,tian2019capacity,zhu2019new}.
A notable example of a scheme that is not of \newName{} is that in~\cite{Sun2018conj}, as discussed in Appendix~\ref{app:counterExample}.

In general, the goal of information-theoretic private information retrieval with $t$-collusion is for the user to retrieve a file such that any set of $t$ storage servers learns nothing about the index of the desired file. This is referred to as \emph{user privacy}.
\begin{definition}[User Privacy with $t$-Collusion]
  Any $t$ colluding servers shall not be able to obtain any information about the index of the requested file, \emph{i.e.}, the mutual information
  \begin{equation}\label{eq:userpriv}
    I(i ;Q_\cT^{i},A_\cT^{i},Y_\cT,S) = 0, \quad \forall \ \cT \subset [n], |\cT| = t \ .
  \end{equation}
\end{definition}

We will also consider symmetric PIR (SPIR), where the user is not supposed to learn any information about the files other than the requested one.
\begin{definition}[Server Privacy]\label{def:symmetricPIR}
  The user shall learn no information about files other than the requested one, \emph{i.e.},
  \begin{equation}\label{eq: serverpriv}
    I(X^{[m]\setminus i} ; A^i,\cQ,i) = 0 \ .
  \end{equation}
\end{definition}
A scheme that satisfies \eqref{eq:correctness} and \eqref{eq:userpriv} is called a PIR scheme. If the scheme in addition satisfies Definition~\ref{def:symmetricPIR}, then it is called an  SPIR scheme.
We are interested in the capacities of linear PIR and SPIR with collusion and adversaries, \emph{i.e.}, the highest achievable rate at which a desired file can be retrieved under these constraints.

\begin{definition}[(S)PIR Rate and Capacity]
  The {\em rate} of an (S)PIR scheme is the number of information bits of the requested file retrieved per downloaded answer bits, \emph{i.e.},
  \begin{equation*}
    R_{\mathrm{(S)PIR}} = \frac{H(X^i)}{\sum_{j=1}^n H(A_j^i)} \ .
  \end{equation*}
\end{definition}

In order to achieve symmetric privacy, the servers require some amount of shared randomness~\cite{gertner2000protecting}.
\begin{definition}[Secrecy Rate]
  The secrecy rate is the amount of common randomness shared by the storage servers relative to the file size, \emph{i.e.},
  \begin{equation*}
    \rho_\spir = \frac{H(S)}{H(X^i)} \ .
  \end{equation*}
\end{definition}

We give some results closely related to the ones presented in this work.
\begin{theorem}[Capacity of TSPIR~{\cite[Theorem~1]{Wang2017col}} and TBSPIR~{\cite[Theorem~1]{Wang2017adv}}]
  For linear symmetric private information retrieval from a set of $m\geq 2$ files stored on $n$ servers with an $(n,k)$ MDS code (for replication $k=1$), where any $t$ servers may collude, the capacity is
    \begin{equation*}
    C_{\tspir}^{(n,k)-\mds} = \left\{ \begin{array}{ll} 1-\frac{k+t-1}{n}, & \mathrm{if} \;\rho_\tspir^{(n,k)-\mds} \geq \frac{k+t-1}{n-k-t+1}\\ 0, & \mathrm{otherwise} \end{array} \right. .
  \end{equation*}
  For symmetric private information retrieval from a set of $m\geq 2$ files replicated on $n$ servers, where any $t$ servers may collude and any $b$ servers are Byzantine, the capacity is
  \begin{equation*}
      C_{\tbspir}^{[n,1]-\mds} = \left\{ \begin{array}{ll} 1-\frac{2b+t}{n}, & \mathrm{if} \;\rho_\tbspir \geq \frac{t}{n-t-2b}\\ 0, & \mathrm{otherwise} \end{array} \right. .
  \end{equation*}
\end{theorem}

It is known that when $t=1$  or $k=1$, the above SPIR capacity coincides with the asymptotic capacity of PIR with no server privacy \cite{Banawan2018, Sun2016}. Motivated by this, our aim is to prove that this is the case more generally. Namely, we will prove Conjectures~2 and~3 for MDS-coded linear PIR, and then proceed to provide a proof in the case of strongly-linear schemes (see Definition~\ref{def:strongly}) for Conjecture~1, further extending the conjectured asymptotic capacity to the nonasymptotic regime in this special case.

In our proofs, we will repeatedly use Han's inequality for joint entropies~\cite{Cover:1991}, which we state here for completeness. Let $W=\{W_1,\dots W_n\}$ be a set of random variables defined on the same probability space.
Denote by $\binom{[n]}{k}$ the set of all subsets of $[n]$ with cardinality $k$. Then
\begin{equation}\label{eq:han}
\frac{k}{n}H(W_1,\dots W_n)\leq \frac{1}{\binom{n}{k}}\sum_{\mathcal{T}\in\binom{[n]}{k}} H(W_\mathcal{T}).
\end{equation}

\section{Preliminary Lemmas}

We begin by introducing some intermediate lemmas which will be required in both Section~\ref{sec:capa} and Section~\ref{sec:symmetricCap}.
Our proofs of linear \newName{} MDS-TPIR in Section~\ref{sec:capa} and MDS-TBSPIR capacity in Section~\ref{sec:symmetricCap} are partly based on the proofs of TBSPIR capacity in a replicated setting \cite{Wang2017adv} as well as the proofs of SPIR capacity~\cite{Wang2016} and TSPIR capacity~\cite{Wang2017col} from MDS-coded storage. We first prove the intermediate results for a set of servers that is free of adversaries and then, similar to~\cite{Wang2017adv}, argue that the entropy of the adversarial responses has to be the same as for nonadversarial servers to obtain the capacity. For completeness, the proofs of the intermediate steps are included, though some of the proofs can be taken directly from~\cite{Wang2017col} and~\cite{Wang2017adv}.

Similar to the replicated case in~\cite[Lemma 6]{Wang2017adv}, in the following we argue that when considering zero error probability, \emph{i.e.}, guaranteeing that the user can decode if the number of corrupted answers is less than or equal to $b$ and the number of nonresponsive servers is less than or equal to $r$, every realization of $n-2b-r$ authentic answers has to be unique.

\begin{lemma}\label{lem:DecodeFromN2B}
  In an optimal scheme with zero error probability for $b$ adversarial and $r$ nonresponsive servers it holds that
  \begin{equation*}
    H(X^i | A_{\mathcal{H}}^i , \cQ) = 0 \ ,
  \end{equation*}
  for any set $\mathcal{H}$ of honest servers with $|\mathcal{H}| \geq n-2b-r$.
\end{lemma}
\begin{proof}
  The proof is similar to the replicated case of \cite[Lemma~6]{Wang2017adv} and included for completeness.
  We show that the response of any $n-2b-r$ honest servers must suffice to correctly recover the desired file by proving that the corresponding responses must be unique for any realization of file $i$. Denote by $A_j^{i}(X^i=x^i)$ the \emph{honest} response of the $j^{\rm th}$ server for the realization $X^i=x^i$ of the $i^{\rm th}$ file. For a contradiction, assume that for a set $\cR \subset [n]$ with $|\cR| = r$ of nonresponsive servers and a set $\cH \subset [n] \setminus \cR$ of honest servers with $|\cH| = n-2b-r$ it holds that $A_{\cH}^i(X^i=x^i) = A_{\cH}^i(X^i=\tilde{x}^i)$ for two different realizations $x^i \neq \tilde{x}^i$ of file $i$.
  Partition the $2b$ remaining servers $\mathcal{B} = [n]\setminus (\cH \cup \cR)$ into two subsets $\cB_1$ and $\cB_2$, each of size $b$, and denote their responses by $A_{\mathcal{B}_1}^i$ and $A_{\mathcal{B}_2}^i$, respectively. Now consider the following cases:
  \begin{itemize}
  \item The realization of file $i$ is $X^i = x^i$. The servers of $\cB_1$ are \emph{adversarial} and reply with $A_{\cB_1}^i(X^i = \tilde x^i)$. The servers of $\cB_2$ are \emph{honest}, \emph{i.e.}, they reply with $A_{\cB_2}^i(X^i = x^i)$.
  \item The realization of file $i$ is $X^i = \tilde x^i$. The servers of $\cB_1$ are \emph{honest}, \emph{i.e.}, they reply with $A_{\cB_1}^i(X^i = \tilde x^i)$. The servers of $\cB_2$ are \emph{adversarial} and reply with $A_{\cB_2}^i(X^i = x^i)$.
  \end{itemize}
  As $A_{\cH}^i(X^i=x^i) = A_{\cH}^i(X^i=\tilde{x}^i)$ by assumption, the user receives exactly the same responses from the servers in both cases and is therefore not able to differentiate between the two realizations. Hence unique decoding would fail, thereby violating the zero error probability requirement. Note that, as we require the \emph{zero} decoding error probability it is not necessary for the adversarial servers to know the index $i$. Instead, in each case it suffices that the probability of the adversarial servers replying with the respective responses is nonzero. We conclude that for any two different realizations $x^i \neq \tilde{x}^i$ of file $i$ we have $A_{\cH}^i(X^i=x^i) \neq A_{\cH}^i(X^i=\tilde{x}^i)$, and the statement of the lemma follows.
\end{proof}

The following basic lemma will also be required in multiple proofs and applies to both the symmetric and nonsymmetric setting.
\begin{lemma}\label{lem:AnswerQueryDependence}
  For any set $\cN \subset [n]$ of nonadversarial servers
  \begin{equation*}
    H(A_{\cN}^i | \cQ, X^i, Q_{\cN}^i ) = H(A_{\cN}^i | X^i, Q_{\cN}^i) .
  \end{equation*}
\end{lemma}
\begin{proof}
  We first show that $I(A_{\cN}^i ; \cQ | X^i, Q_{\cN}^i) \leq 0$, as follows
  \begin{align*}
    I(A_{\cN}^i ; \cQ &| X^i, Q_{\cN}^i) \leq I(A_{\cN}^i, X^{[m]}, S ; \cQ | X^i, Q_{\cN}^i)\\
                                    &\stackrel{(a)}{=} I(X^{[m]} , S; \cQ | X^i, Q_{\cN}^i)\\
    &\stackrel{\phantom{(\sfa)}}{=} H( X^{[m]} , S | X^i, Q_{\cN}^i) - H( X^{[m]} , S | X^i, Q_{\cN}^i, \cQ) \\
    &\stackrel{(b)}{=} H( X^{[m]} , S | X^i) - H( X^{[m]} , S | X^i) = 0 \ ,\\
  \end{align*}
  where $(a)$ follows because the answers $A_{\cN}^i$ are a function of the queries $Q_{\cN}^i$, the files $X^{[m]}$, and the shared randomness $S$ (for the nonsymmetric case $S$ can be thought of as a constant, \emph{e.g.}, $S=\mathbf{0}$), and $(b)$ holds because the files $X^{[m]}$ and shared randomness $S$ are independent of the queries. As mutual information is nonnegative, it follows that
  \begin{align*}
    I(A_{\cN}^i ; \cQ | X^i, Q_{\cN}^i) = &H(A_{\cN}^i | X^i, Q_{\cN}^i) \\
                                       &- H(A_{\cN}^i | \cQ , X^i, Q_{\cN}^i) = 0\\
    \Rightarrow H(A_{\cN}^i | X^i, Q_{\cN}^i) = &H(A_{\cN}^i | \cQ , X^i, Q_{\cN}^i) \ .
  \end{align*}
\end{proof}

\section{Proof of Linear \newName{} MDS-TPIR Capacity}
\label{sec:capa}

\subsection{Converse}
\label{subsec:converse}

A novel formulation of the key Lemma~\ref{keylemma}, which is slightly stronger than the corresponding lemmas in~\cite{Wang2016, Wang2017col}, allows us to induct over the number of files, without requiring the symmetry assumption. We then use this induction result to prove the MDS-TPIR capacity for linear, \newName{} schemes. The same proof also yields an upper bound for the capacity in the presence of adversarial servers. However, the upper bound for MDS-TBPIR does not correspond to any known scheme constructions, and does not even agree with the MDS-TBSPIR capacity asymptotically as the number of files grows to infinity.

The following lemma, which is key to our capacity bounds, describes how sets of as many as $k+t-1$ servers will give responses that are independent of the index of the desired file, even when conditioned on an arbitrary subset of files. In order to show this we need some additional technical results on the rank of the Khatri-Rao product \cite{khatri1968solutions} of certain matrices. To not disturb the flow of the paper, we give these (lengthy) statements and corresponding proofs in Appendix~\ref{app:technicalLemmas}.
As we are only concerned with nonsymmetric PIR here, we assume $\bS = \mathbf{0}$ for the remainder of this section.
\begin{lemma}\label{lem:Nsets}
  Let $\cN \subset [n]$ with $n>k+t-1$ be a set of $|\cN| = k+t-1$ nonadversarial servers, and let $\cF\subsetneq [m]$ be any proper subset of the rows of the MDS-coded storage system. For any optimal linear, \newName{} PIR scheme, and any $i, i'\in[m]$, it holds that \begin{equation}\label{eq:answers}
    H(A_{\cN}^i | X^{\cF}, Q_{\cN}^i) = H(A_{\cN}^{i'} | X^{\cF}, Q_{\cN}^{i'}) \ .
  \end{equation}
\end{lemma}
\begin{proof}
First note that an equivalent problem formulation\footnote{We choose to refer to the realizations of $\cQ_\cN^i$ as $\bq[:,\psi_{\beta}(\cN)]$ to be consistent with notation and indexing, \emph{i.e.}, we treat the realizations of $\cQ_\cN^i$ as a submatrix consisting of $k+t-1$ thick columns of the realizations $\bq$ of $\cQ^i$.} is given by (cf. \cite[Section~2.2, Eq. (2.10)]{Cover:1991})
  \begin{align*}
    &\underset{\bq[:,\psi_{\beta}(\cN)] \in \supp(Q^i_\cN)}{\mathbb{E}} \big(H(A_{\cN}^i | X^{\cF},Q_{\cN}^i=\bq[:,\psi_{\beta}(\cN)] )\big) \\
    &= \underset{\bq[:,\psi_{\beta}(\cN)] \in \supp(Q^{i'}_\cN)}{\mathbb{E}} \big(H(A_{\cN}^i | X^{\cF}, Q_{\cN}^{i'}=\bq[:,\psi_{\beta}(\cN)])\big) \ .
  \end{align*}
  Further, observe that by Definition~\ref{def:linearPIR} the responses of servers $j\in [\cN]$ can be expressed as the star-product (Hadamard product) between rows of the restricted query matrix $\bQ^i[:,\psi_\beta(\cN)]$ and the restricted storage matrix with each column repeated $\beta$ times, \emph{i.e.}, $\bY \otimes \mathbf{1}_\beta$, where $\otimes$ denotes the Kronecker product. Specifically, we have
  \begin{align}
      A_\cN^i &= \bA^i[:,\psi_\beta(\cN)] \nonumber \\
      &= \big(\myspan{\bY[:,j],\bQ^i[:,(j-1)\beta+s]} \big)_{s\in [\beta], j \in \cN} \nonumber \\
      &= \big( \bY[:,j]^\top \cdot \bQ^i[:,\psi_\beta(j)] \big)_{j\in \cN} \nonumber \\
      &= \sum_{l \in \psi_{\alpha}([m])} \big((\bY[l,j] \otimes \mathbf{1}_\beta) \star \bQ^i[l,\psi_{\beta}(j)] \big)_{j\in \cN} \nonumber \\
      &= \sum_{l \in \psi_{\alpha}([m])} \big((\bY[l,\cN] \otimes \mathbf{1}_\beta) \star \bQ^i[l,\psi_{\beta}(\cN)] \big) \ . \label{eq:responseAsSP}
  \end{align}
  Next, we show that for every query realization, the entropies only depend on the size of the support of the query realization as
  \begin{align*}
    H&(A_{\cN}^i | X^{\psi_{\alpha}(\cF)},Q_{\cN}^i=\bq[:,\psi_{\beta}(\cN)] ) \\
     &= H\Big( \sum_{l \in \psi_{\alpha}([m])} \big((\bY[l,\cN] \otimes \mathbf{1}_\beta) \star \bq[l,\psi_{\beta}(\cN)]\ \big) \  \\
    &\hspace{4.5cm} \Big| \ X^\cF ,Q_{\cN}^i=\bq[:,\psi_{\beta}(\cN)] \Big) \\
     &= H\Big( \sum_{l \in \psi_{\alpha}([m]\setminus \cF)} \big((\bY[l,\cN] \otimes \mathbf{1}_\beta) \star \bq[l,\psi_{\beta}(\cN)]\ \big) \ \\
    &\hspace{5.2cm} \Big| \ Q_{\cN}^i=\bq[:,\psi_{\beta}(\cN)]\Big) \\
    &\stackrel{(\sfa)}{=} \big|\colsupp\big(\bq[ \psi_{\alpha}([m]\setminus \cF) ,\psi_{\beta}(\cN)]\big)\big| \ ,
  \end{align*}
  where $(\sfa)$ holds by Lemma~\ref{lem:starProductDimensionBoundQueryMDS} given in Appendix~\ref{app:technicalLemmas}.
  Taking the expectation over the support of $Q^i_\cN$ gives
  \begin{align*}
    &\underset{\bq[:,\psi_{\beta}(\cN)] \in \supp(Q^{i}_\cN) }{\mathbb{E}} \big(H(A_{\cN}^i | X^{\cF},Q_{\cN}^i=\mathbf{q}[:,\psi_{\beta}(\cN)] )\big)\\
    &\stackrel{\phantom{(\sfa)}}{=} \underset{\bq[:,\psi_{\beta}(\cN)] \in \supp(Q^{i}_\cN) }{\mathbb{E}} \big|\colsupp\big(\bq[ \psi_{\alpha}([m]\setminus \cF) ,\psi_{\beta}(\cN)]\big)\big| \\
    &\stackrel{(\sfa)}{=} \underset{\bq[:,\psi_{\beta}(\cN)] \in \supp(Q^{i'}_\cN) }{\mathbb{E}} \big|\colsupp\big(\bq[ \psi_{\alpha}([m]\setminus \cF) ,\psi_{\beta}(\cN)]\big)\big|\\
    &\stackrel{\phantom{(\sfa)}}{=} \underset{\bq[:,\psi_{\beta}(\cN)] \in \supp(Q^{i'}_\cN)}{\mathbb{E}} H(A_{\cN}^{i'} | X^{\cF}, Q_{\cN}^{i'}=\bq[:,\psi_{\beta}(\cN)]) \ ,
  \end{align*}
  where $(\sfa)$ follows from Lemma~\ref{lem:expectationSupport} given in Appendix~\ref{app:technicalLemmas}.
\end{proof}

While the previous lemma holds for any pair of indices $i,i' \in [m]$, the interesting case is when $i\in \cF$, $i'\not\in \cF$. This is the case that intuitively means that $(k+t-1)$-tuples of servers handle desired and undesired files equally, and will be used in the inductive proof of Lemma~\ref{keylemma}. Also note that the property of \newName{} was needed in the proof of Lemma~\ref{lem:rankKhatri}, the key technical ingredient to the proof of Lemma~\ref{lem:starProductDimensionBoundQueryMDS} and thereby also to Lemma~\ref{lem:Nsets}, as it ensures that the given entropy expression is equal to the size of the column support of the query restricted to the respective rows and columns.

\begin{remark}\label{rem:newPIRformulation}
The formulation of the server responses used in Lemma~\ref{lem:Nsets} implies a novel formulation of the PIR problem with linear decoding functions. As shown in $\eqref{eq:vectorize}$ and Lemma~\ref{lem:starProductDimensionBoundQueryMDS}, the received responses are given by (to simplify the notation we assume $\alpha=1$ here)
\begin{align}
  \bA^i &= (\bX[1,1],\bX[2,1],\ldots, \bX[m,1],\bX[1,2],\bX[2,2], \nonumber\\
  &\qquad \ldots, \bX[m,2], \ldots, \bX[m,k]) \cdot \big((\bG \otimes \mathbf{1}_\beta) \odot \bQ^i\big) \ , \label{eq:responsesKhatri}
\end{align}
where $\odot$ denotes the column-wise Khatri-Rao product \cite{khatri1968solutions} and $\bG$ is a generator matrix of the storage code.
When restricting to linear decoding functions, the application of a decoder $\mathsf{D}$ such that $\sfD(\bA) = \bX^i$, is equivalent to performing linear combinations of the received responses $\bA^i$, which, in turn, is equivalent to performing linear combinations of the columns of $(\bG \otimes \mathbf{1}_\beta) \odot \bQ^i$. It is easy to see that the $l^{\mathrm{th}}$ symbol of the information vector can be obtained exactly if the $l^{\mathrm{th}}$ unit vector $e_l$ is in the column span of this matrix. Therefore, the problem of linear PIR with linear decoding functions can be defined solely based on operations from linear algebra: For each $i\in[m]$ determine a distribution of query matrices $\bQ^i$ such that
\begin{align*}
  &e_l \in \myspanCol{(\bG \otimes \mathbf{1}_\beta) \odot \bQ^i} \\
  &\hspace{2.5cm} \forall \ l \in \{i,m+i,2m+i,\ldots, (k-1)m+i\}\\
   &\Pr(\bQ^i_\cT = \bq) = \Pr(\bQ^{i'}_\cT = \bq)  \ \forall \ i,i'\in [m] , \cT \subset [n], |\cT| \leq t \ .
\end{align*}
The first condition guarantees decodability, as the given set indexes the symbols of file $i$ in the data vector of \eqref{eq:responsesKhatri}, while the second condition guarantees $t$-privacy.
\end{remark}

The following lemma will be used to prove the upper bounds on the nonsymmetric MDS-TPIR capacity.

\begin{lemma}\label{keylemma}
Consider an optimal linear (S)PIR scheme, and let $\mathcal{H} \subset [n]$ be a minimal set (set of smallest possible cardinality) such that the requested file $i$ can be obtained from the respective responses, \emph{i.e.},
\begin{equation*}
    H(X^i | A_{\mathcal{H}}^i , \cQ) = 0 \ .
\end{equation*}
For $1\leq s\leq m$, let
\begin{align*}
h_{s}=\frac{n}{|\mathcal{H}|} H(A^{s}_{\mathcal{H}}|Q,X^{1,\dots ,s-1})
\end{align*}
and $h_{m+1}=0$. Then, for all $1\leq s\leq m$,
\begin{align*}
h_{s}\geq \frac{n}{n-2b-r}\left(H(X^{s})+\frac{k+t-1}{n}h_{s+1}\right)\ .
\end{align*}
\end{lemma}

\begin{proof}%
By Lemma~\ref{lem:DecodeFromN2B}, $|\mathcal{H}|\leq n-2b-r$.
By Han's inequality~\eqref{eq:han}, the average value of $H(A^{s+1}_\mathcal{N}|Q,X^{1,\dots ,s})$ over all sets $\cN \subseteq \mathcal{H}$ with $|\cN| = k+t-1$ is at least
\begin{align*}
\frac{k+t-1}{|\mathcal{H}|} H(A^{s+1}_\mathcal{H}|Q,X^{1,\dots ,s}) \ .
\end{align*}
We can thus choose a set $\cN \subseteq \mathcal{H}$ with $|\cN| = k+t-1$ such that
\begin{align*}\label{NinH}
H(A^{s+1}_\mathcal{N}|Q,X^{1,\dots ,s}) &\geq \frac{k+t-1}{|\mathcal{H}|} H(A^{s+1}_\mathcal{H}|Q,X^{1,\dots ,s})\\
&= \frac{k+t-1}{n} h_{s+1}\ .
\end{align*}

By independence of the files and the queries, we have
\begin{align*}
H(X^{s}|Q,X^{1,\dots, s-1})=H(X^s)\ .
\end{align*}
We thus get
\begin{align*}
 h_s & = \frac{n}{|\mathcal{H}|}H(A^{s}_{\mathcal{H}}|Q,X^{1,\dots, s-1}) \\
 	& = \frac{n}{|\mathcal{H}|}\left(H(X^{s}) + H(A^{s}_{\mathcal{H}}|Q,X^{1,\dots ,s})\right) \\
	& \geq \frac{n}{|\mathcal{H}|}\left(H(X^{s}) + H(A^{s}_{\mathcal{N}}|Q,X^{1,\dots ,s})\right) \\
	& \stackrel{(\sfa)}{=}  \frac{n}{|\mathcal{H}|}\left(H(X^{s}) + H(A^{s+1}_{\mathcal{N}}|Q,X^{1,\dots ,s})\right) \\
	& \geq \frac{n}{|\mathcal{H}|} \left(H(X^s) + \frac{k+t-1}{n}h_{s+1}\right)\\
	&\geq \frac{n}{n-2b-r} \left(H(X^s) + \frac{k+t-1}{n}h_{s+1}\right)\ ,
 \end{align*}
 where $(\sfa)$ follows from Lemma~\ref{lem:Nsets}.
  \end{proof}

Setting $b=r=0$, we are now ready to prove the capacity of MDS-TPIR. We restate Theorem~\ref{thm:coded-colluded} here for the sake of completeness.

\mainthm*

\begin{proof}
\emph{Achievability:} An explicit scheme achieving the rate is constructed in~\cite{oliveira2019one} by ``lifting'' the star product scheme of \cite{Freij-Hollanti2017}. To be private, this scheme needs to fulfill Definition~\ref{def:newPIRproperty}, as discussed in Appendix~\ref{app:liftedFix}.

\emph{Converse:} Let $\mathcal{H} \subset [n]$ be a minimal set such that \begin{equation*}
   H(X^i | A_{\mathcal{H}}^i , \cQ) = 0,  \end{equation*} and for $s=1,\dots, m$, let
\begin{align*}
h_s=\frac{n}{|\mathcal{H}|} H(A^{s}_{\mathcal{H}}|Q,X^{1,\dots ,s-1})
\end{align*}
as in Lemma~\ref{keylemma}.
Denote the size $H(X^i)$, which is equal for all files, by $L$. By definition and Lemma~\ref{lem:AnswerQueryDependence}, the rate $R$ of the scheme satisfies
\begin{align*}
  \frac{1}{R}&=\frac{\sum_{j\in [n]} H(A_j^s)}{L}\\
             &\geq \frac{\sum_{j\in[n]}H(A^{s}_j|Q)}{L} \\
  &\geq \frac{h_1}{L} \ ,
\end{align*}
where the last equation follows by minimality of $\cH$.
It is thus enough to show that
\begin{equation}\label{eq:induction}
  \frac{h_s}{L}\geq \frac{1-\left(\frac{k+t-1}{n} \right)^{m-s+1}}{1-\frac{k+t-1}{n}}
\end{equation}
holds for all $\ 1\leq s\leq m$. We will prove this by backwards induction on $s$.

As the base case consider $s=m$ and observe that~\eqref{eq:induction} simplifies to $h_s\geq L$ in this case. Recall that $A^m$ is a function of the files $X$ and the queries $Q$. As we have $b=r=0$, Lemma~\ref{keylemma} gives
  \begin{align*}
    h_m &= H(A^{m}_{\mathcal{H}}|Q,X^{1,\dots ,m-1}) \\
    &= H(A^{m}_{\mathcal{H}}, X^m |Q,X^{1,\dots ,m-1}) = H(X^m) = L \ .
  \end{align*}
 It follows that \eqref{eq:induction} is correct for $s=m$.
Now assume as an induction hypothesis that
\begin{align*}
\frac{h_{s'}}{L}\geq \frac{1-\left(\frac{k+t-1}{n}\right)^{m-s'+1}}{1-\frac{k+t-1}{n}} \ ,
\end{align*}
and let $s=s'-1$
Then Lemma~\ref{keylemma} yields
\begin{align*}
  \frac{h_{s}}{L}&\geq 1+\left(\frac{k+t-1}{n}\frac{h_{s'}}{L}\right)\\
& \geq 1+\frac{\frac{k+t-1}{n}-\left(\frac{k+t-1}{n}\right)^{m-s'+2}}{1-\frac{k+t-1}{n}}\\
& = \frac{1-\left(\frac{k+t-1}{n}\right)^{m-s+1}}{1-\frac{k+t-1}{n}} \ .
\end{align*}
This proves~\eqref{eq:induction} for all $\ 1\leq s\leq m$ by induction. The case $s=1$ is the statement of the theorem.
\end{proof}

\begin{remark}
  By similar techniques, we get an upper bound
  \begin{equation}\label{eq:TBbound}
    C_{\tbpir}^{[n,k]-\mds}\leq\left(1-\frac{2b+r}{n}\right) \cdot \frac{1-\frac{k+t-1}{n}}{1-\left(\frac{k+t-1}{n}\right)^m}
  \end{equation}
  for the case where we also have Byzantine and nonresponsive servers. However, we believe this to be a loose upper bound.
If the bound~\eqref{eq:TBbound} were to be tight, the result given in Theorem~\ref{thm:symmetricCapacity} would imply that in this setting symmetric PIR has a strictly lower capacity than PIR even as the number of files goes to infinity. This would be in sharp contrast to the known cases of TPIR/TSPIR and MDS-PIR/MDS-SPIR with and without Byzantine/nonresponsive servers, where the nonsymmetric capacity converges (from above) to the symmetric capacity as the number of files increases. %
\end{remark}

\section{Strongly-linear PIR Capacity}\label{sec:strongly}

We have seen that, for a symmetric linear scheme, the rate cannot be larger than that obtained by a star product scheme in~\cite{Tajeddine2018}, regardless of the number of files. Further, Theorem~\ref{thm:coded-colluded} shows that as the number of files grows, the rate of the star product scheme in~\cite{Freij-Hollanti2017} approaches the \newName{} capacity. We will now show that, under stronger linearity assumptions, this is also true for a finite number of files and without assuming server privacy. In essence, we define a {\em strongly linear} PIR scheme to be one where all interference cancellation is linear and deterministic, and where every computation uses only one response symbol from each server. This is a highly natural assumption, that also has practical implications as it allows decoding to happen instantly, even when queries are sent sequentially. However, the assumption is not true for schemes such as those in \cite{Sun2016, Banawan2018}, which do not satisfy Definition~\ref{def:strongly} below. %

\begin{definition}[Strongly Linear PIR]
\label{def:strongly}
We say that a linear PIR scheme is \emph{strongly linear} if each symbol of the desired file is obtained as a deterministic linear function over $\F$ of a response vector consisting of one response symbol from each server
\begin{align*}
\big(\bA[:,(j-1)\beta+s]\big)_{j \in [n]} \ ,
\end{align*} for some $s\in[\beta]$. %
By this, we mean that the reconstruction function does not depend on the randomness used to produce the queries. We informally think of $\big(\bA[:,(j-1)\beta+s]\big)_{j \in [n]}$ as the response obtained in the $s^\mathrm{th}$ iteration of the scheme. %
\end{definition}

\begin{remark}
Note that a \newName{} PIR scheme does not have to be strongly linear. However, the rate of every optimal strongly linear scheme is upper bounded by the rate of the star product scheme~\cite{Freij-Hollanti2017}, which agrees with the asymptotic capacity of a \newName{} PIR scheme  with corresponding parameters. This result is proved in Theorem~\ref{thm:strongly}. Hence, a \newName{} scheme can always be replaced by a strongly linear scheme (\emph{e.g.}, a star product scheme) without a loss in the asymptotic rate.
\end{remark}

For the results in this section, we need to recall a notion that is central to much recent work on PIR. For two vectors $\bc, \bd \in \F^{1 \times n}$ their coordinate-wise/star-/Hadamard product is denoted
\begin{align*}
\bc\star \bd =(\bc[1,1] \bd[1,1] , \dots , \bc[1,n] \bd[1,n]) \ .
\end{align*}
Let $\cC$ and $\cD$ be two codes of length $n$ over $\F$. The star product $\cC\star \cD$ is the linear span of the codewords $\bc\star \bd$, where $\bc\in \cC$, $\bd\in \cD$. Note that this definition does not require that the codes $\cC$ and $\cD$ are linear, but it always yields a linear code $\cC\star \cD$ as the star product.

\begin{lemma}\label{lem:strongly}
Consider a strongly linear PIR scheme from a linear storage code $\mathcal{C}$, and fix an index $s\in[\beta]$. For all $l\in [m]$, let $\mathcal{D}^{i,l}\subseteq \F^n$ be the linear span of the row vectors that can occur as the $l^{\rm th}$ %
row of the $s^{\rm th}$ iteration of a query matrix $Q^i$, \emph{i.e.},
\begin{align*}
\mathcal{D}^{i,l} \!= \! \myspanRow{\bQ[l,\{s, \beta\!+\!s, \dots , (n\!-\!1)\beta\!+\!s\}]  : \bQ^i \!\in\!\supp(Q^i)}.
\end{align*}
Then the rate of the PIR scheme is at most
  \begin{align*}
1-\frac{\dim(\mathcal{C}\star(\sum_{l\not\in \psi_{\alpha}(i)} \mathcal{D}^{i,l}))}{\dim(\mathcal{C}\star(\sum_{l} \mathcal{D}^{i,l}))} \ .
  \end{align*}
If the strongly linear PIR scheme downloads equally much from all servers, then the rate is at most
\begin{align*}
\frac{\dim(\mathcal{C}\star(\sum_{j} \mathcal{D}^{i,j}))-\dim(\mathcal{C}\star(\sum_{l\not\in \psi_{\alpha}(i)} \mathcal{D}^{i,j}))}{n} \ .
\end{align*}
\end{lemma}
\begin{proof}
By \eqref{eq:responseAsSP} the responses in a linear PIR scheme as in Definition~\ref{def:linearPIR} can be described as the sum of the star product (\emph{i.e.}, Hadamard product) of rows of the query matrix and rows of the storage by
\begin{align*}
  \big(\bA[:,(j-1)\beta+s]\big)_{j \in [n]} &= \sum_{l=1}^{\alpha m} \big(\bY[l,(j-1)\beta+s]\big)_{j \in [n]} \\
  &\hspace{1cm}\star \bQ^i[l,(j-1)\beta+s]\big)_{j \in [n]} \big)\\
      &\in \sum_{l=1}^{\alpha m} \mathcal{D}^{i,l} \star \mathcal{C} \ .
\end{align*}
Let
\begin{align*}
    \Phi:\big(\bA[:,(j-1)\beta+s]\big)_{j \in [n]}\mapsto \bx\in\F^\gamma
\end{align*}
be the deterministic map that returns  $\gamma$ desired coordinates %
of the desired file $\bX$ from the responses in iteration $s$. Then for each $l\not\in\psi_{\alpha}(i)$, $\Phi$ must be constant on each coset of $\mathcal{D}^{i,l} \star \mathcal{C}$, because otherwise changing the query matrix and the $l^{\rm th}$ row of $\bY$ would affect the value of $\Phi\big(\big(\bA[:,(j-1)\beta+s]\big)_{j \in [n]}\big)$. Since this holds for every $l\neq i$, $\Phi$ must be constant on each coset of $\sum_{j\neq i} \mathcal{D}^{i,j}\star \mathcal{C}$. Thus, the dimension of the range of $\Phi$ is
\begin{align*}
    \gamma&=\dim \Big(\sum_{j} \mathcal{D}^{i,j}\star \mathcal{C}\Big)-\dim\ker(\Phi)\\
    &\leq\dim\Big(\sum_{j} \mathcal{D}^{i,j}\star \mathcal{C}\Big)-\dim\Big(\sum_{j\not\in\psi_\alpha(i)} \mathcal{D}^{i,j}\star \mathcal{C}\Big).
\end{align*}
The answer $\big(\bA[:,(j-1)\beta+s]\big)_{j \in [n]}$ can be reconstructed from the responses of $\dim\left(\sum_{j} \mathcal{D}^{i,j}\star \mathcal{C}\right)$ servers, or from $n$ servers if we require to download equally much from each server. Dividing the number $|\mathcal{I}|$ of downloaded $q$-ary symbols from the desired file by the number of $q$-ary symbols in $\big(\bA[:,(j-1)\beta+s]\big)_{j \in [n]}$, we get the claimed bounds on the PIR rate. This concludes the proof.
\end{proof}

For the rest of the paper, we assume downloading the same number of symbols from all the servers for simplicity.

Before proceeding, for the reader's convenience, we briefly recapitulate the star product PIR scheme of \cite{Freij-Hollanti2017,Tajeddine2018}. Consider a distributed storage system storing $m$ files encoded with an $[n,k]$ MDS storage code $\cC$ and a user looking to retrieve file $i$ with collusion resistance $t$. For simplicity we assume $n=2k+t+2b+r-1$ and $\alpha=1$ here, as this allows the recovery of the file in one iteration\footnote{We would like to emphasize that the scheme discussed here is a special case of the star product PIR scheme of \cite{Freij-Hollanti2017,Tajeddine2018}, with parameters chosen for an illustrative purpose. The full scheme is not limited to this specific choice of $n$.}, for more details see \cite{Freij-Hollanti2017,Tajeddine2018}. Further, for ease of notation, we only consider the case of all servers being responsive, \emph{i.e.}, $r=0$. The extension to the case of nonresponsive servers is trivial. The star product scheme consists of the following steps:
  \begin{enumerate}
  \item The user chooses a \emph{query code} $\cD_Q$ with $d_{\cD_Q^{\perp}}\geq t+1$, where $d_{\cD_Q^{\perp}}$ denotes the minimum distance of the dual code $\cD_Q^{\perp}$. From this code, she generates a matrix $\bD \in F^{m\times n}$ whose $m$ rows are codewords of $\cD_Q$ chosen i.i.d. at random\footnote{The fact that $d_{\cD_Q^\perp}\geq t+1$ implies that any $t$ positions in a codeword of $\cD_Q$ are an information set. Hence, any $t$ columns of $\bD$ are i.i.d. distributed over $\F^{m \times t}$.}.
  \item The \emph{query matrix} is given by
    \begin{align*}
      \bQ^i = \bD + \bE \ ,
    \end{align*}
    where $\bE$ is all-zero, except for the $i^{\mathrm{th}}$ row $\bE[i,:]$, which is chosen to be the basis of an $[n,1]$ code\footnote{Here and for the general scheme it is convenient to view this as a code instead of a vector. Note that for a different choice of $n$ and $\alpha$ the dimension of this code could be larger than $1$.} $\mathcal{E}$.
  \item The user sends the $j^{\rm th}$ column of $\bQ^i$ to the $j^{\rm th}$ server. The server replies with $\bA^i[1,j] = \left\langle \bQ^i[:,j], \bY[:,j]\right\rangle + \bz[1,j]$, where $\bz[1,j]=0$ if the server is honest and arbitrary if the server is adversarial ($\bz$ can be thought of as the received error vector).
  \item By \eqref{eq:responseAsSP} the user receives
    \begin{align*}
      \bA^i &= \Big(\sum_{l \in [m]} \bY[l,:] \star \bQ^i[l,:] \Big) + \bz \\
            &= \Big(\sum_{l \in [m]} \bY[l,:] \star \big(\bD[l,:]+\bE[l,:]\big) \Big) + \bz \\
            &= \underbrace{\Big(\sum_{l \in [m]} \bY[l,:] \star \bD[l,:] \Big)}_{\in \cC \star \cD_Q } + \underbrace{\big(\bY[i,:] \star \bE[i,:]\big)}_{\in \cC \star \mathcal{E}} + \bz \\
    \end{align*}
    Recall that the Hamming weight of $\bz$ is at most $b$, the number of adversarial servers. Hence, if the code $\cC \star \cD_Q + \cC \star  \mathcal{E}$ is of distance $d_{\cC\star \cD_Q + \cC\star  \mathcal{E}}\geq 2b+1$, the errors can be decoded and the user obtains
    \begin{align*}
            \underbrace{\Big(\sum_{l \in [m]} \bY[l,:] \star \bD[l,:] \Big)}_{\in \cC\star \cD_Q } + \underbrace{\big(\bY[i,:] \star \bE[i,:]\big)}_{\in \cC\star \mathcal{E} } \ .
    \end{align*}
    As $\bE$ is chosen by the user, we only require that the codes $\cC \star \cD_Q $ and $\cC \star \mathcal{E} $ intersect trivially to recover the vector $\bY[i,:]\star \bE[i,:] \in \cC \star \mathcal{E}$. Finally, the file $X^i$ can be recovered from this vector, given that $\cC \star \mathcal{E} $ is of dimension $k$.
  \end{enumerate}
  It remains to determine codes $\cC$, $\cD_Q$, and $\mathcal{E}$ that fulfill the required properties for the given $n$. Conveniently, it has been shown \cite{Freij-Hollanti2017,Tajeddine2018} that the popular class of generalized Reed-Solomon (GRS) codes provides such codes, however, these details are beyond the scope of this short summary.

We are now ready to show that any strongly linear scheme can be replaced by a star product scheme for the same privacy model, without losing in the PIR rate. %
\begin{theorem}[Capacity of Strongly Linear PIR]\label{thm:strongly}
  The capacity of strongly linear PIR from an $(n,k)$ storage code $\cC$, with $b$ Byzantine and $r$ non responsive servers, that protects against $t$-collusion, is
\begin{equation*}
    C_{\tbpir}^{(n,k)-\mds}= 1-\frac{k+t+2b+r-1}{n} \,
\end{equation*}
for any number of files $m$.
\end{theorem}
\begin{proof} Consider an arbitrary strongly linear PIR scheme. Like in Lemma~\ref{lem:strongly}, fix an iteration $s\in[\beta]$ and define
  \begin{align*}
\mathcal{D}^{i,l} \!=\! \myspan{(\bQ^i[l,\{s, \beta\!+\!s, \dots , (n\!-\!1)\beta\!+\!s\}]): \bQ^i\! \in\!\supp(Q^i)}
  \end{align*}
  for $l\in[n]$.  Define $\mathcal{D}=\sum_{l\neq i}\mathcal{D}^{i,l}$. Let $\mathbf{E}\in\F^{\alpha m\times  n}$ be a matrix such that $\mathbf{E}[\psi_\alpha(i), :]$ is an arbitrary realisation of $\bQ^i[\psi_\alpha(i),\{s, s+\beta, \dots , s+(n-1)\beta\}]$, and all other entries are zero. Let $\mathbf{D}\in\F^{\alpha m\times  n}$ be a random matrix whose rows are selected uniformly at random from $\mathcal{D}$.

Now consider the star product scheme with query matrix $\mathbf{D}+ \mathbf{E}$. %
This scheme has a set of feasible query matrices that is more restrictive in the row of the desired file, but less restrictive in the rows of the unwanted files, than the strongly linear scheme under consideration. Thus, whatever privacy constraints were satisfied by the original scheme, including robustness against non-responsive and byzantine servers, are also respected by the star product scheme. By design all symbols that were decoded in the $r^\mathrm{th}$ iteration of the strongly linear scheme are also decoded in the star product scheme. Moreover, by construction the rate of the star product scheme is
\begin{align*}
1-\frac{\dim(\mathcal{D})}{\dim(\sum_{l}\mathcal{D}^{i,l})} \ ,
\end{align*}
which is at least the rate of the original strongly linear scheme by Lemma~\ref{lem:strongly}. So the rate of any strongly linear scheme is bounded from above by the rate of a star product scheme with the same privacy constraints, which is in turn bounded by $1-\frac{k+2b+r+t-1}{n}$ as shown in \cite{Tajeddine2018}. The paper also presents a scheme achieving this bound via the star product construction.
\end{proof}

Note that the capacity of strongly linear PIR is \emph{independent} of the number of files (see also the remark below). Hence, the above theorem also yields a proof for Conjecture~1 in the strongly linear case. The capacity of a strongly linear scheme also matches the asymptotic rate of Conjecture~3, hence proving the asymptotic expression for such schemes.

\begin{remark}
Here, to simplify the notation, we have assumed that all the servers respond with equal size responses. However, by loosening this assumption, improvements for finite $m$ are possible, along the same lines as in \cite{zhou2020capacity}. %
The proof of the above theorem shows that, among strongly-linear schemes as in Definition~\ref{def:strongly}, the star product scheme \cite{Freij-Hollanti2017, Tajeddine2018} is optimal.
\end{remark}

\section{Capacity of MDS-coded TBSPIR for schemes with additive randomness}
\label{sec:symmetricCap}

  In this section we prove the capacity of MDS-coded TBSPIR for the specific system models considered in \cite{Wang2019symmetric}. Recent works \cite{Wang2019isit,Wang2019symmetric} have shown that it is crucial to consider the distribution of the randomness shared by the servers when deriving the capacity of such systems. We begin by shortly reviewing the results presented in these works. In \cite{Wang2019isit} the authors derive the capacity of MDS-coded SPIR with mismatched randomness, meaning that they assume the complete randomness to be available to all servers. It is shown that this assumption of sharing the complete randomness among the servers leads to a strictly larger rate than when the randomness is also coded with the MDS storage code, referred to as \emph{matched randomness}. The resulting capacity approaches the capacity of coded, matched SPIR when the number of files tends to infinity and is always strictly lower than the coded PIR capacity.

In \cite{Wang2019symmetric} the authors derive the capacity of MDS-coded SPIR with and without collusion for the case of matched randomness, \emph{i.e.}, where the randomness is also encoded with the storage code. Further, they consider the special case of schemes with additive randomness independent of the queries. Specifically, the authors show
\begin{itemize}
\item the capacity of $(n,k)$ MDS-coded storage, where for any $k$ servers the randomness is independent, to be
  \begin{equation*}
    C_{\text{matched MDS-SPIR}} = 1-\frac{k}{n} .
  \end{equation*}
\item the capacity of uncoded, \emph{i.e.}, $k=1$, SPIR with collusion of any $t$ servers (TSPIR) to be
  \begin{equation*}
    C_{\text{TSPIR}} = 1-\frac{t}{n} \ .
  \end{equation*}
\item the capacity of MDS-coded TSPIR, for schemes where the servers add the randomness to the responses and the randomness is independent of the queries to be
  \begin{equation*}
    C_{\text{add. MDS-TSPIR}} = 1-\frac{k+t-1}{n}\ .
  \end{equation*}
\end{itemize}

In this section we consider the extension of the results from \cite{Wang2019symmetric} to the MDS-TBSPIR setting, \emph{i.e.}, to symmetric PIR from coded databases protecting against a number of byzantine servers $b$ and a number of unresponsive servers $r$. %

\begin{definition}[Matched BSPIR \cite{Wang2019symmetric}]\label{def:matchedSPIR}
We say a BSPIR scheme is matched if the randomness shared by the servers is independent for every subset of $k$ servers.
\end{definition}

\begin{definition}[Additive randomness TBSPIR \cite{Wang2019symmetric}]\label{def:additiveRandomness}
We define a scheme to be an \emph{additive randomness} TBSPIR scheme if the responses are of the form
\begin{equation*}
    A_j^i = f_j(Q_j^i,Y_j) + S_j
\end{equation*}
where $S_j$ is independent of the received query $Q_j^i$.
\end{definition}

\begin{lemma}\label{lem:DropFileTBSPIR}
  For any MDS-TBSPIR scheme and for any set of nonadversarial servers $\cN \subset [n]$ with $|\cN| = k+t-1$ it holds that
  \begin{equation*}
      H(A_\cN^i | X^i, Q_\cN^i) = H(A_\cN^{i}|Q_\cN^i) \ ,
  \end{equation*}
  if the randomness is additive as in Definition~\ref{def:additiveRandomness} or $t=1$.
\end{lemma}
\begin{proof}
  The proof for the case of additive randomness follows directly from the proof of \cite[Lemma~8]{Wang2019symmetric}, as it is independent of the total number of servers and we restrict the lemma to nonadversarial servers. For the same reasons the proof of \cite[Lemma~7]{Wang2019symmetric} also applies here for the case of $t=1$.
\end{proof}

\begin{theorem}\label{thm:symmetricCapacity}
  The capacity of linear symmetric PIR from $[n,k]$ MDS-coded storage with $b$ adversarial, $r$ nonresponsive and $t$ colluding servers is given by
  \begin{equation*}
    1-\frac{k+t+2b+r-1}{n} \ ,
  \end{equation*}
    if the randomness is additive as defined in Definition~\ref{def:additiveRandomness} or for $t=1$ as in Definition~\ref{def:matchedSPIR}.
\end{theorem}

\begin{proof}
  Let $\mathcal{H} \subset [n]$ and $\cN \subset \mathcal{H}$ be sets of honest, responsive servers with $|\mathcal{H}| = n-2b-r$ and $|\cN| = k+t-1$. Then
  \begin{align*}
    H(X^i) &\stackrel{(a)}{=} H(X^i | \cQ) \\
                &\stackrel{(b)}{=} H(X^i | \cQ) - H(X^i | A_{\mathcal{H}}^i, \cQ) \\
                &= I(X^i ; A_{\mathcal{H}}^i | \cQ) \\
                &= H(A_{\mathcal{H}}^i| \cQ) - H(A_{\mathcal{H}}^i | X^i , \cQ) \\
                &\stackrel{(c)}{\leq} H(A_{\mathcal{H}}^i | \cQ) - H(A_\cN^i | X^i , \cQ, Q_\cN^i)\\
                &\stackrel{(d)}{=} H(A_{\mathcal{H}}^i | \cQ) - H(A_\cN^i | X^i, Q_\cN^i)\\
                &\stackrel{(e)}{=} H(A_{\mathcal{H}}^i | \cQ) - H(A_\cN^{i} |  Q_\cN^{i})\\
                &\leq H(A_{\mathcal{H}}^i | \cQ) - H(A_\cN^{i} |  \cQ) \ ,
  \end{align*}
  Equality $(a)$ holds because the files are independent of the queries, $(b)$ holds by Lemma~\ref{lem:DecodeFromN2B}, $(c)$ holds because $\cN \subset \mathcal{H}$, $(d)$ holds by Lemma~\ref{lem:AnswerQueryDependence}, and $(e)$ holds by Lemma~\ref{lem:DropFileTBSPIR}.

  Averaging over all sets $\cN$ gives
  \begin{equation*}
    H(X^i) \leq H(A_{\mathcal{H}}^i | \cQ) - \frac{1}{\binom{n-2b-r}{k+t-1}} \sum_{\substack{\cN \subset \mathcal{H} \\ |\cN| = k+t-1}} H(A_\cN^i | \cQ)
  \end{equation*}
  and by Han's inequality (see Equation~(\ref{eq:han})) %
  \begin{equation*}
    \frac{1}{\binom{n-2b-r}{k+t-1}} \sum_{\substack{\cN \subset \mathcal{H} \\ |\cN| = k+t-1}} H(A_\cN^i | \cQ) \geq \frac{k+t-1}{n-2b-r} H(A_{\mathcal{H}}^i | \cQ) .
  \end{equation*}
  Hence, there exists an $h \in \mathcal{H}$ such that
  \begin{align}
    H(X^i) &\leq H(A_{\mathcal{H}}^i | \cQ) - \frac{k+t-1}{n-2b-r} H(A_{\mathcal{H}}^i | \cQ) \nonumber \\
           &= \frac{n-k-2b-r-t+1}{n-2b-r} H(A_{\mathcal{H}}^i | \cQ) \nonumber \\
                &\leq \frac{n-k-2b-r-t+1}{n-2b-r} (n-2b-r) H(A_h^i | \cQ) \ . \nonumber
  \end{align}
  Since the adversaries could otherwise be easily identified, we can assume that the answers of the adversarial servers are of the same entropy as the nonadversarial answers. This gives
  \begin{align}
    \frac{H(X^i)}{\sum_{j=1}^n H(A_j^i)} &= \frac{H(X^i) }{ n \cdot H(A_h^i | \cQ)} \label{eq:symCapacitySum} \\
                                         &\leq \frac{H(X^i)}{n} \cdot \frac{n-k-2b-r-t+1}{H(X^i)} \nonumber \\
    &= \frac{n-k-2b-r-t+1}{n} \ . \nonumber
  \end{align}
\emph{Achievability:} The symmetric version of the scheme introduced in \cite{Tajeddine2018}, which generalizes the scheme of \cite{Freij-Hollanti2017}, achieves the presented upper bound on the PIR rate. Note that this scheme fulfills both Definition~\ref{def:matchedSPIR} and Definition~\ref{def:additiveRandomness}, since the symmetry is achieved by adding a random codeword from the $(n,k)$ MDS storage code to the answers.
\end{proof}

Note that we include the nonresponsive servers in the calculation of the download cost, which is debatable, due to the reasonable argument that nonresponsive servers do not contribute to this cost. However, this depends on the particular system as, \emph{e.g.}, dropped packets on the side of the user could also cause a missing response, while clearly causing network traffic. Therefore we include the nonresponsive servers in the download cost, but note that this can be modified by changing the upper limit of the sum in~\eqref{eq:symCapacitySum} to~$n-r$.

Finally, we derive the secrecy rate of TBSPIR by combining the proofs of~\cite[Theorem~7]{Wang2017col} and~\cite[Theorem~1]{Wang2017adv}.
\begin{theorem}
  The secrecy rate of a linear TBSPIR scheme from an~$(n,k)$ MDS-coded storage system fulfills
  \begin{equation*}
    \rho \geq \frac{k+t-1}{n-k-t-2b-r+1} \ ,
  \end{equation*}
  if the randomness is additive as in Definition~\ref{def:additiveRandomness} \textbf{or} $t=1$.
\end{theorem}
\begin{proof}
\emph
 Let $\mathcal{H} \subset [n]$ and $\cN \subset \mathcal{H}$ be sets of honest, responsive servers with $|\mathcal{H}| = n-2b-r$ and $|\cN| = k+t-1$.
First, observe that
\begin{align*}
H(A_{\mathcal{H}}^i | \cQ) & = H(X^i) + H(A_{\mathcal{H}}^i|X^i, Q)\\
& \geq H(X^i) + H(A_{\mathcal{N}}^i|X^i, Q)\\
& \geq H(X^i) + H(A_{\mathcal{N}}^i|Q).
\end{align*}
 Averaging over all sets $\cN \subset \mathcal{H}$ with  $|\cN| = k+t-1$ we get \begin{equation}\label{eq:symHH}
 H(A_{\mathcal{H}}^i | \cQ)  \geq H(X^i) + \frac{k+t-1}{n-2b-r}H(A_{\mathcal{H}}^i|Q).
 \end{equation}
   Let $\mathcal{H} \subset [n]$ and $\cN \subset \mathcal{H}$ be sets of honest, responsive servers with $|\mathcal{H}| = n-2b-r$ and $|\cN| = k+t-1$. By server privacy,
  \begin{align*}
    0 &= I(X^{[m]\setminus i}; A_{\cH} | \cQ)\\
      &= H(X^{[m]\setminus i} | \cQ) - H(X^{[m]\setminus i} | A_{\cH}^i, \cQ)\\
      &= H(X^{[m]\setminus i} | X^i, \cQ) - H(X^{[m]\setminus i} | A_{\cH}^i, X^i, \cQ)\\
      &= I(X^{[m]\setminus i}; A_{\cH}^i| \cQ, X^i)\\
      &\geq I(X^{[m]\setminus i}; A_{\cN}^i| \cQ, X^i)\\
      &= H(A_{\cN}^i | X^i, \cQ) - H(A_{\cN}^i | X^{[m]}, \cQ) + H(A_{\cN}^i | S, X^{[m]}, \cQ) \\
      &= H(A_{\cN}^i | X^i, \cQ) - I(S; A_{\cN}^i | X^{[m]}, \cQ) \\
      &\geq H(A_{\cN}^i | X^i, Q_{\cN}^i, \cQ) - H(S) \\
      &= H(A_{\cN}^i | Q_{\cN}^i) - H(S)\\
      &\geq H(A_{\cN}^i | \cQ) - H(S) \,.
  \end{align*}
  Averaging over all sets $\cN$, we get by~\eqref{eq:symHH} that
  \begin{align*}
    H(S) &\geq \frac{1}{\binom{n-2b-r}{k+t-1}} \sum_{\substack{\cN \subset \cH \\ |\cN| = k+t-1}} H(A_{\cN}^i | \cQ)\\
         &\geq \frac{k+t-1}{n-2b-r} H(A_{\cH}^i | \cQ)\\
    &\geq \frac{k+t-1}{n-k-t-2b-r+1} H(X^i) \ ,
  \end{align*}
    The bound on the secrecy rate follows by
  \begin{equation*}
    \rho = \frac{H(S)}{H(X^i)} \geq \frac{k+t-1}{n-k-t-2b-r+1}\ .
  \end{equation*}
\end{proof}

\section{Optimality of the Star Product Scheme for some parameters} \label{sec:subpacketization}

\subsection{Replicated, no Collusion/Adversaries}
In the previous sections we derived the capacity for some specific settings. In the case of linear, \newName{} MDS-TPIR, as considered in Section~\ref{sec:capa}, this capacity depends on the number of files in the system, similar to the capacity in the known settings of PIR from replicated storage \cite{Sun2017}, MDS-PIR \cite{Banawan2018}, and TPIR \cite{Sun2016}. For a finite number of files, the rate of PIR schemes in these settings can be increased compared to the asymptotic (as the number of files $m \rightarrow \infty$) capacity, as restated in Table~\ref{tab:capacity}. However, whether this improvement can actually be realized depends on the level of subpacketization $L = \alpha k$, \emph{i.e.}, the number of symbols in a file, the size of the alphabet on which the PIR scheme operates, and the size of the alphabet used for transmission\footnote{The download from each server is made up of an integer number of symbols from the transmission alphabet.}. In this section, we derive the explicit relation between these system parameters for which this asymptotic regime is reached. Thereby, we show that for some parameter regimes strongly linear schemes, and therefore also the star-product scheme of \cite{Freij-Hollanti2017,Tajeddine2018}, are in fact optimal in terms of rate.

The following results are based on the results for binary schemes in~\cite{shah2014one}, where it was proved that if $L\leq n-1$, the optimal download is $L+1$ bits. Now consider a scheme over a $q$-ary alphabet and transmission using a $q'$-ary alphabet.  In \cite{sun2017optimal} the results of \cite{shah2014one} were generalized to arbitrary alphabets, by including conditions on the size of the alphabet used for transmission of data to the user. For the case of mismatched alphabets, \emph{i.e.}, $q\neq q'$, the optimal download cost is determined up to a constant offset and for the case of \emph{matched alphabets}, \emph{i.e.}, $q=q'$, a complete characterization of the optimal download cost is given. In particular, \cite{sun2017optimal} shows that the result on the optimal download of \cite{shah2014one} is a special case of their result, \emph{i.e.}, for a subpacketization of $L\leq n-1$ symbols of a $q$-ary alphabet, the optimal download over a $q$-ary alphabet is $L+1$ symbols.

Denote by $D_L$ the download cost for a given level of subpacketization $L$, and define it as the maximum number of symbols
\begin{equation*}
  D_L = \max \sum_{j=1}^n |A_j^i|_{M'}
\end{equation*}
of the transmission alphabet $M'$ a user has to download for any realization of the queries
In \cite{sun2017optimal} it was shown that the optimal download for PIR from $n$ noncolluding databases, each storing all $m\geq 2$ files, for message size $L$ is given by
\begin{equation} \label{eq:optimalDownload}
  D_L = \ceil{\frac{L}{C}} \ ,
\end{equation}
where $C$ is the capacity of unrestrained PIR given by
\begin{equation*}
  C = \frac{1-\frac{1}{n}}{1-\frac{1}{n^m}} \ .
\end{equation*}

\subsection{Coded Storage with Collusion}

We generalize this approach to the case of PIR from coded databases and/or colluding servers.
In Section~\ref{sec:capa} the capacity of linear, \newName{} PIR was shown to be
\begin{equation}\label{eq:capacitySeparated}
  C = \frac{1-\frac{k+t-1}{n}}{1-\left(\frac{k+t-1}{n}\right)^m} \ .
\end{equation}

\begin{remark} \label{rem:capacityExpression}
This expression does not hold in full generality, \emph{i.e.}, when not assuming linearity and \newName{}, see \cite{Sun2018conj} for a counter-example for the case of $m=2$ files. However, it is the best rate for which a scheme for general parameters is known, given by applying the technique presented in \cite{oliveira2019one} to the PIR scheme of \cite{Freij-Hollanti2017}. Further, note that \eqref{eq:capacitySeparated} includes both, the capacity for uncoded storage ($k=1$) with collusion ($t\geq 1$) \cite{Sun2016} and the capacity for coded storage ($k\geq 1$) without collusion ($t=1$) \cite{Banawan2018} as special cases, both of which are valid in general. Hence using this expression provides some insight beyond the linear, \newName{} case.
\end{remark}

It is easy to see that the proof of converse for \eqref{eq:optimalDownload} given in \cite[Section~IV]{sun2017optimal} for the case of matched alphabets ($q=q'$) also applies in this setting. Our goal in the following is not to characterize the optimal download cost, but instead to determine the number of files required to reach the asymptotic regime for a given set of parameters.

Consider a scheme that achieves the asymptotic capacity (cf. Table~\ref{tab:capacity}), such as, \emph{e.g.}, the star product PIR scheme \cite{Freij-Hollanti2017}.
Such a PIR scheme can obtain $n-k-t+1$ symbols of the desired file by downloading $n$ symbols, one from each server. Assume a subpacketization of
\begin{equation*}
    L = \beta (n-k-t+1) \ .
\end{equation*}
A file can be obtained privately by applying this PIR scheme $\beta$ times, downloading a total of $\beta n$ symbols
For the optimal download cost for linear, \newName{} PIR (cf. Remark~\ref{rem:capacityExpression} for motivation of using this expression) we obtain
\begin{align*}
  D_L &= \ceil{\frac{L}{C}} = \ceil{\frac{\beta(n-k-t+1)}{(1-\frac{k+t-1}{n})\left(1-\left(\frac{k+t-1}{n}\right)^{m}\right)^{-1} }}\\
      &= \ceil{\beta n\left(1-\left(\frac{k+t-1}{n}\right)^{m}\right)}\\
      &= \beta n-\floor{\beta n\left(\frac{k+t-1}{n}\right)^{m}} \ .
\end{align*}
Hence, in the linear, \newName{} setting and under the assumption that the transmission alphabet equals the alphabet of the PIR scheme, the asymptotic regime is reached when
\begin{equation*}
  \floor{\beta n\left(\frac{k+t-1}{n}\right)^{m}} = 0
\end{equation*}
which is equivalent to
\begin{align*}
\beta < \frac{n^{m-1}}{(k+t-1)^m} \ ,
\end{align*} or in terms of the number of files,
\begin{align*}
m>\frac{\log n + \log \beta}{\log n-\log (k+t-1)}\ .
\end{align*}
On the other hand, if the scheme has optimal rate, then only $k$ symbols have been downloaded from each stripe, so $L$ is a multiple of $k$. Thus we get
\begin{align*}
  \beta=\frac{L}{n-k-t+1}\geq \frac{k}{\gcd(n-k-t+1,k)}%
  \ .
\end{align*}
For example, for parameters $n=30$, $k=15$, $t=10$, and $\beta = 5$ this condition if fulfilled for $m\geq 23$ files.

For the simplest case of replicated storage ($k=1$) and no collusion ($t=1$) this can be further simplified to show that the asymptotic regime is reached when
\begin{align*}
    \beta < n^{m-1} \ .
\end{align*}
Note that for the nontrivial settings of $m\geq 2$ this is always fulfilled if $\beta < n$.

\section{Conclusions and Future Work}
In this paper, we defined the practical notions of \newName{} PIR and strongly linear PIR. We have proved the capacity of MDS-coded, linear, \newName{} PIR with colluding servers. The capacity of symmetric linear PIR with MDS-coded, colluding, Byzantine, and nonresponsive servers was proved for the case of matched randomness.

The results on \newName{} PIR are a significant step towards the general proof for MDS-coded and colluded PIR capacity. Meanwhile, the results on strongly linear PIR bear high practical interest in that these schemes allow for small field sizes and low subpacketization levels, making implementation of the schemes much simpler. These simpler schemes also achieve the same asymptotic capacity as the \newName{} schemes.
The main open problem that remains is proving the capacity of (linear) PIR with MDS-coded and colluding servers without the assumption of \newName{}. As explained in Section~\ref{sec:introduction} the presented definition of \newName{} PIR isolates a property required for a scheme to achieve this capacity for general linear, MDS-coded PIR, namely for the restrictions of its queries to \emph{not} be of full support-rank. Thereby, the results in this paper  provide a good starting point for both giving upper bounds on the PIR rate and constructing achieving schemes.

Another open problem is determining the capacity of TPIR for transitive storage codes, along the lines of \cite{Freij-Hollanti2019transitive}, by adapting the proofs of Lemma~\ref{lem:rankKhatri} and~\ref{lem:starProductDimensionBoundQueryMDS} accordingly.

\section*{Acknowledgment}
The authors would like to thank Prof.~Syed Jafar, Prof.~Chao Tian, and Dr. Qiwen Wang for helpful discussions.

\bibliographystyle{IEEEtran}
\bibliography{paperSymCap}

\begin{thebibliography}{10}
\providecommand{\url}[1]{#1}
\csname url@samestyle\endcsname
\providecommand{\newblock}{\relax}
\providecommand{\bibinfo}[2]{#2}
\providecommand{\BIBentrySTDinterwordspacing}{\spaceskip=0pt\relax}
\providecommand{\BIBentryALTinterwordstretchfactor}{4}
\providecommand{\BIBentryALTinterwordspacing}{\spaceskip=\fontdimen2\font plus
\BIBentryALTinterwordstretchfactor\fontdimen3\font minus
  \fontdimen4\font\relax}
\providecommand{\BIBforeignlanguage}[2]{{%
\expandafter\ifx\csname l@#1\endcsname\relax
\typeout{** WARNING: IEEEtran.bst: No hyphenation pattern has been}%
\typeout{** loaded for the language `#1'. Using the pattern for}%
\typeout{** the default language instead.}%
\else
\language=\csname l@#1\endcsname
\fi
#2}}
\providecommand{\BIBdecl}{\relax}
\BIBdecl

\bibitem{Holzbaur2019}
L.~Holzbaur, R.~Freij-Hollanti, and C.~Hollanti, ``On the capacity of private
  information retrieval from coded, colluding, and adversarial servers,'' in
  \emph{2019 IEEE Information Theory Workshop (ITW)}, Aug. 2019.

\bibitem{chor1995private}
B.~Chor, O.~Goldreich, E.~Kushilevitz, and M.~Sudan, ``Private information
  retrieval,'' in \emph{Proceedings of IEEE 36th Annual Foundations of Computer
  Science}.\hskip 1em plus 0.5em minus 0.4em\relax IEEE, 1995, pp. 41--50.

\bibitem{Sun2017}
H.~Sun and S.~A. Jafar, ``The capacity of private information retrieval,''
  \emph{IEEE Transactions on Information Theory}, vol.~63, no.~7, pp.
  4075--4088, jul 2017.

\bibitem{Banawan2018}
K.~Banawan and S.~Ulukus, ``The capacity of private information retrieval from
  coded databases,'' \emph{IEEE Transactions on Information Theory}, vol.~64,
  no.~3, pp. 1945--1956, mar 2018.

\bibitem{Sun2016}
H.~{Sun} and S.~A. {Jafar}, ``The capacity of robust private information
  retrieval with colluding databases,'' \emph{IEEE Transactions on Information
  Theory}, vol.~64, no.~4, pp. 2361--2370, April 2018.

\bibitem{Heidarzadeh2018single}
A.~{Heidarzadeh}, F.~{Kazemi}, and A.~{Sprintson}, ``Capacity of single-server
  single-message private information retrieval with coded side information,''
  in \emph{2018 IEEE Information Theory Workshop (ITW)}, Nov 2018, pp. 1--5.

\bibitem{Heidarzadeh2018singlemulti}
A.~{Heidarzadeh}, B.~{Garcia}, S.~{Kadhe}, S.~E. {Rouayheb}, and
  A.~{Sprintson}, ``On the capacity of single-server multi-message private
  information retrieval with side information,'' in \emph{2018 56th Annual
  Allerton Conference on Communication, Control, and Computing}, Oct 2018, pp.
  180--187.

\bibitem{sun2018sym}
H.~Sun and S.~A. Jafar, ``The capacity of symmetric private information
  retrieval,'' \emph{IEEE Transactions on Information Theory}, vol.~65, no.~1,
  pp. 322--329, 2018.

\bibitem{Wang2017adv}
Q.~{Wang} and M.~{Skoglund}, ``Secure symmetric private information retrieval
  from colluding databases with adversaries,'' in \emph{2017 55th Annual
  Allerton Conference on Communication, Control, and Computing (Allerton)}, Oct
  2017, pp. 1083--1090.

\bibitem{Wang2017col}
------, ``Linear symmetric private information retrieval for {MDS} coded
  distributed storage with colluding servers,'' in \emph{2017 IEEE Information
  Theory Workshop (ITW)}, Nov 2017, pp. 71--75.

\bibitem{Wang2016}
------, ``Symmetric private information retrieval for {MDS} coded distributed
  storage,'' in \emph{2017 IEEE International Conference on Communications
  (ICC)}, May 2017, pp. 1--6.

\bibitem{Wang2019symmetric}
------, ``Symmetric private information retrieval from {MDS} coded distributed
  storage with non-colluding and colluding servers,'' \emph{IEEE Transactions
  on Information Theory}, vol.~65, no.~8, pp. 5160--5175, Aug 2019.

\bibitem{Freij-Hollanti2019transitive}
R.~{Freij-Hollanti}, O.~W. {Gnilke}, C.~{Hollanti}, A.~{Horlemann-Trautmann},
  D.~{Karpuk}, and I.~{Kubjas}, ``$t$-private information retrieval schemes
  using transitive codes,'' \emph{IEEE Transactions on Information Theory},
  vol.~65, no.~4, pp. 2107--2118, April 2019.

\bibitem{Kumar2019arblin}
S.~{Kumar}, H.~{Lin}, E.~{Rosnes}, and A.~G. i.~{Amat}, ``Achieving maximum
  distance separable private information retrieval capacity with linear
  codes,'' \emph{IEEE Transactions on Information Theory}, pp. 1--1, 2019.

\bibitem{Freij-Hollanti2017}
R.~Freij-Hollanti, O.~W. Gnilke, C.~Hollanti, and D.~A. Karpuk, ``Private
  information retrieval from coded databases with colluding servers,''
  \emph{SIAM Journal on Applied Algebra and Geometry}, vol.~1, no.~1, pp.
  647--664, jan 2017.

\bibitem{Tajeddine2018}
R.~Tajeddine, O.~W. Gnilke, D.~Karpuk, R.~Freij-Hollanti, and C.~Hollanti,
  ``Private information retrieval from coded storage systems with colluding,
  {B}yzantine, and unresponsive servers,'' \emph{IEEE Transactions on
  Information Theory}, vol.~65, no.~6, pp. 3898--3906, 2019.

\bibitem{Sun2018conj}
H.~{Sun} and S.~A. {Jafar}, ``Private information retrieval from {MDS} coded
  data with colluding servers: Settling a conjecture by {Freij-Hollanti} et
  al.'' \emph{IEEE Transactions on Information Theory}, vol.~64, no.~2, pp.
  1000--1022, Feb 2018.

\bibitem{Banawan2019byzcoll}
K.~{Banawan} and S.~{Ulukus}, ``The capacity of private information retrieval
  from byzantine and colluding databases,'' \emph{IEEE Transactions on
  Information Theory}, vol.~65, no.~2, pp. 1206--1219, Feb 2019.

\bibitem{oliveira2019one}
R.~G. D'Oliveira and S.~El~Rouayheb, ``One-shot {PIR}: Refinement and
  lifting,'' \emph{IEEE Transactions on Information Theory}, vol.~66, no.~4,
  pp. 2443--2455, 2019.

\bibitem{tajeddine2018private}
R.~Tajeddine, O.~W. Gnilke, and S.~El~Rouayheb, ``Private information retrieval
  from {MDS} coded data in distributed storage systems,'' \emph{IEEE
  Transactions on Information Theory}, vol.~64, no.~11, pp. 7081--7093, 2018.

\bibitem{zhang2018optimal}
Z.~Zhang and J.~Xu, ``The optimal sub-packetization of linear
  capacity-achieving {PIR} schemes with colluding servers,'' \emph{IEEE
  Transactions on Information Theory}, vol.~65, no.~5, pp. 2723--2735, 2018.

\bibitem{zhou2020capacity}
R.~Zhou, C.~Tian, H.~Sun, and T.~Liu, ``Capacity-achieving private information
  retrieval codes from {MDS}-coded databases with minimum message size,''
  \emph{IEEE Transactions on Information Theory}, vol.~66, no.~8, pp.
  4904--4916, 2020.

\bibitem{sun2017optimal}
H.~Sun and S.~A. Jafar, ``Optimal download cost of private information
  retrieval for arbitrary message length,'' \emph{IEEE Transactions on
  Information Forensics and Security}, vol.~12, no.~12, pp. 2920--2932, 2017.

\bibitem{gertner2000protecting}
Y.~Gertner, Y.~Ishai, E.~Kushilevitz, and T.~Malkin, ``Protecting data privacy
  in private information retrieval schemes,'' \emph{Journal of Computer and
  System Sciences}, vol.~60, no.~3, pp. 592--629, 2000.

\bibitem{li2020towards}
J.~Li, D.~Karpuk, and C.~Hollanti, ``Towards practical private information
  retrieval from {MDS} array codes,'' \emph{IEEE Transactions on
  Communications}, vol.~68, no.~6, pp. 3415--3425, 2020.

\bibitem{tian2019capacity}
C.~Tian, H.~Sun, and J.~Chen, ``Capacity-achieving private information
  retrieval codes with optimal message size and upload cost,'' \emph{IEEE
  Transactions on Information Theory}, vol.~65, no.~11, pp. 7613--7627, 2019.

\bibitem{zhu2019new}
J.~Zhu, Q.~Yan, C.~Qi, and X.~Tang, ``A new capacity-achieving private
  information retrieval scheme with (almost) optimal file length for coded
  servers,'' \emph{IEEE Transactions on Information Forensics and Security},
  vol.~15, pp. 1248--1260, 2019.

\bibitem{Cover:1991}
T.~M. Cover and J.~A. Thomas, \emph{Elements of Information Theory}.\hskip 1em
  plus 0.5em minus 0.4em\relax New York, NY, USA: Wiley-Interscience, 1991.

\bibitem{khatri1968solutions}
C.~Khatri and C.~R. Rao, ``Solutions to some functional equations and their
  applications to characterization of probability distributions,''
  \emph{Sankhy{\=a}: The Indian Journal of Statistics, Series A}, pp. 167--180,
  1968.

\bibitem{Wang2019isit}
Q.~Wang, H.~Sun, and M.~Skoglund, ``Symmetric private information retrieval
  with mismatched coded messages and randomness,'' in \emph{2019 IEEE
  International Symposium on Information Theory (ISIT)}, July 2019.

\bibitem{shah2014one}
N.~B. Shah, K.~Rashmi, and K.~Ramchandran, ``One extra bit of download ensures
  perfectly private information retrieval,'' in \emph{2014 IEEE International
  Symposium on Information Theory}.\hskip 1em plus 0.5em minus 0.4em\relax
  IEEE, 2014, pp. 856--860.

\bibitem{Slyusar97}
V.~Slyusar, ``New operations of matrices product for applications of radars,''
  in \emph{Proc. Direct and Inverse Problems of Electromagnetic and Acoustic
  Wave Theory}, 1997, pp. 73--74.

\bibitem{Horn91}
R.~A. Horn and C.~R. Johnson, \emph{Matrix Analysis}.\hskip 1em plus 0.5em
  minus 0.4em\relax Cambridge University Press, 1991.

\bibitem{randriambololona2013upper}
H.~Randriambololona, ``An upper bound of {Singleton} type for componentwise
  products of linear codes,'' \emph{IEEE Transactions on Information Theory},
  vol.~59, no.~12, pp. 7936--7939, 2013.

\end{thebibliography}

\appendices

\section{Properties of \newName{} PIR Schemes}\label{app:technicalLemmas}

We define the following notation of matrix products:
\begin{itemize}
\item[$\cdot$] the regular matrix/vector/scalar product (if obvious from context, we neglect this symbol) \\ $\F^{m\times n}\times \F^{n\times n'}\to\F^{m\times n'}$
\item[$\star$] the star-product / Hadamard product \\ $\F^{m\times n}\times \F^{m\times n}\to\F^{m\times n}$
\item[$\otimes$] the Kronecker product \\ $\F^{m\times n}\times \F^{m'\times n'}\to\F^{mm'\times nn'}$
\item[$\odot$] the column-wise Khatri-Rao product  \cite{khatri1968solutions} \\ $\F^{m\times n}\times \F^{m'\times n}\to\F^{mm'\times n}$
\item[$*$] the row-wise Khatri-Rao product / face-splitting product  \cite{Slyusar97,khatri1968solutions} \\ $\F^{m\times n}\times \F^{m\times n'}\to\F^{m\times nn'}$
\end{itemize}
It is easy to check (see, \emph{e.g.}, \cite{Slyusar97} and \cite[Lemma~4.2.10.]{Horn91}) that for matrices $\bA,\bB,\bC,\bD$ and a row vector $\bz$ it holds that
\begin{align}
  (\bA\cdot \bB)\star (\bC\cdot \bD) &= (\bA * \bC)\cdot (\bB \odot \bD) \label{eq:mixedProduct}\\
  (\bA \cdot \bB) \otimes \bz &= \bA \cdot (\bB \otimes \bz) \label{eq:assoKronecker}\\
  \myspanRow{\bA} \star \myspanRow{\bB} &= \myspanRow{ \{ \ba \star \bb \ | \ \ba \in \myspanRow{\bA}, \bb \in \myspanRow{\bB}\}} \nonumber \\
  &= \myspanRow{\bA \odot \bB} \ .\label{eq:starProductKhatri}
\end{align}
Further, observe that for a matrix $\bX\in \F^{m \times k}$ we have
\begin{align}
  \mathbf{1}_{m} \cdot &\left(\bX * \bI_{m} \right) = \big(\bX[1,1],\bX[2,1],\ldots, \bX[m,1],\bX[1,2], \nonumber\\
   & \quad \bX[2,2],\ldots, \bX[m,2], \ldots, \bX[m,k]\big) \ \in \F^{1 \times km} \ , \label{eq:vectorize}
\end{align}
where $\bI_{m}$ denotes the $m\times m$ identity matrix. Moreover, if the matrix $\bX$ is uniformly distributed over $\F^{k \times m}$, then $  \mathbf{1}_{m} ( \bX * \bI_{m} )$ is uniformly distributed over $\F^{1 \times km}$.

For completeness, we note that the following results also hold if the thick columns are not all of the same size $\beta$, but instead each consist of a (possibly) different number of columns. However, to not complicate the notation even further, we restrict ourselves to PIR schemes that query each node exactly $\beta$ times, which corresponds to equal sized thick columns, each consisting of $\beta$ columns.

\begin{lemma}\label{lem:rankKhatri}
  Let $\cC$ be an $[k+t-1,k]$ MDS code with generator matrix $\bG\in \F_q^{k \times (k+t-1)}$ and $\bq \in \F_{q}^{\alpha \times \beta (k+t-1)}$ be a matrix such that for any set $\cT\subset [k+t-1]$ with $|\cT|=t$ we have
  \begin{align}\label{eq:rankRestrictionMatrix}
    \rank(\bq[:,\psi_{\beta}(\cT)]) = |\colsupp(\bq[:,\psi_{\beta}(\cT)])| \ .
  \end{align}
  Then
  \begin{align*}
  \rank((\bG\otimes \mathbf{1}_\beta) \odot \bq) = |\colsupp(\bq)|  \ .
  \end{align*}
\end{lemma}
\begin{proof}
  By a similar argument as in \cite[Proof of Lemma~6]{randriambololona2013upper}, we determine the rank of this matrix by proving that the unit vectors $\be_l\in \F^{(k+t-1)\beta}, l\in \colsupp(\bq)$ are contained in the row span of the matrix $(\bG\otimes \mathbf{1}_\beta) \odot \bq$, where $\mathbf{1}_\beta$ denotes the all-one vector of length $\beta$.
  First observe that for any full-rank matrices $\bP_1 \in \F^{k \times k}$ and $\bP_2 \in \F^{\alpha \times \alpha}$ we have
  \begin{align}
    \rank((\bG\otimes & \mathbf{1}_\beta) \odot \bq) = \dim(\myspanRow{(\bG\otimes \mathbf{1}_\beta) \odot \bq}) \nonumber\\
                                               &\stackrel{(\sfa)}{=} \dim(\myspanRow{\bG\otimes \mathbf{1}_\beta} \odot \myspanRow{\bq}) \nonumber\\
                                               &\stackrel{(\sfb)}{=} \dim(\myspanRow{\bP_1\cdot (\bG\otimes \mathbf{1}_\beta)} \odot \myspanRow{\bP_2\cdot \bq}) \nonumber\\
                                               &\stackrel{(\sfc)}{=} \dim(\myspanRow{(\bP_1\cdot \bG)\otimes \mathbf{1}_\beta} \odot \myspanRow{\bP_2\cdot \bq}) \label{eq:stepIllustrated}\\
                                               &\stackrel{\phantom{(\sfc)}}{=}  \rank((\bP_1 \cdot \bG)\otimes \mathbf{1}_\beta) \odot (\bP_2 \cdot \bq))  \ ,\nonumber
  \end{align}
  where $(\sfa)$ follows from \eqref{eq:starProductKhatri}, $(\sfb)$ holds because the left-multiplication by a full-rank matrix does not change the row space, and $(\sfc)$ holds by \eqref{eq:assoKronecker}. To obtain the unit vectors $\be_l \in \F^{(k+t-1)\beta}, l \in \psi_\beta(1) \cap \colsupp(\bq)$, choose $\bP_1$ such that $\bP_1 \bG$ is in systematic form, \emph{i.e.}, its first $k$ columns are an identity matrix. This is always possible, since $\bG$ is the generator matrix of and MDS code. Now consider the set $\cT = \{1,k+1,\ldots,k+t-1\}$ and choose $\bP_2$ such that the submatrix of $\bP_2 \cdot \bq$ consisting of the $t$ columns indexed by $\psi_{\beta}(\cT)$, \emph{i.e.}, the $\beta$-thick columns $\cT$, contain the unit vectors $\be_l\in \F^{t\beta}, l \in \colsupp(\bq[:,\psi_{\beta}(\cT)])$ as rows. Condition (\ref{eq:rankRestrictionMatrix}) guarantees that such a matrix exists.

  \begin{figure*}
    \centering
    \def\x{0.85}%

\begin{tikzpicture}

\coordinate (SLM) at (0,0); %

\draw (SLM) -- ++(0.2*\x,1.6*\x);
\draw (SLM) -- ++(0.2*\x,-1.6*\x);

\coordinate (SRM) at ($(SLM)+(15.5*\x,0*\x)$);
\draw (SRM) -- ++(-0.2*\x,1.6*\x);
\draw (SRM) -- ++(-0.2*\x,-1.6*\x);
\node[anchor=west] at ($(SRM)+(-0.15*\x,-1.5*\x)$) {$\mathsf{row}$};

\coordinate (Gbeta) at ($(SLM)+(1*\x,0*\x)$); %

\node[left delimiter=(,minimum height = \x*1.6cm] at (Gbeta) {};

\coordinate (GNW) at ($(Gbeta)+(0.2*\x,.75*\x)$); %
\draw[fill=TUMBlue!20] (GNW) rectangle ++(2.5*\x,-1.5*\x) node[pos=0.5] () {$\bG$};

\node at ($(Gbeta)+(3.2*\x,0*\x)$) {$\otimes$};

\coordinate (OneNW) at ($(Gbeta)+(3.7*\x,0.25*\x)$); %
\draw[fill=TUMBlue!20] (OneNW) rectangle ++(1.5*\x,-.5*\x) node[pos=0.5] () {$\mathbf{1}_\beta$};

\node[right delimiter=),minimum height = \x*1.6cm] at ($(Gbeta)+(5.2*\x,0*\x)$) {};

\node at ($(SLM)+(7*\x,0*\x)$) {$\odot$};

\coordinate (qNW) at ($(SLM)+(7.5*\x,1.5*\x)$); %

\foreach \i in {0,1,1.5,4,5,5.5,7}{
\draw[draw=none, fill=TUMBlue!20] ($(qNW)+(\i*\x,0*\x)$) rectangle ++(0.5*\x,-3*\x) node[rotate=90,pos=0.5,color=TUMBlue] {\scriptsize$\neq \mathbf{0}$};
}

\node[color=TUMGray,rotate=90] at ($(qNW)+(0.75*\x,-1.5*\x)$) {$\bq[:,\psi_{\beta}(1)]$};
\node[color=TUMGray,rotate=90] at ($(qNW)+(2.25*\x,-1.5*\x)$) {$\bq[:,\psi_{\beta}(2)]$};
\node[color=TUMGray] at ($(qNW)+(3.75*\x,-1.5*\x)$) {$\cdots$};

\draw[dashed, color=TUMGray] ($(qNW)+(1.5*\x,0*\x)$) -- ++(0*\x,-3*\x);
\draw[dashed, color=TUMGray] ($(qNW)+(3*\x,0*\x)$) -- ++(0*\x,-3*\x);
\draw[dashed, color=TUMGray] ($(qNW)+(4.5*\x,0*\x)$) -- ++(0*\x,-3*\x);
\draw[dashed, color=TUMGray] ($(qNW)+(6*\x,0*\x)$) -- ++(0*\x,-3*\x);

\draw (qNW) rectangle ++(7.5*\x,-3*\x);

\node[draw=none] (Lsupp) at ($(qNW)+(3.75*\x,1*\x)$) {$\underset{\text{Here} \ \cT = \{1,4,5\}}{\text{Support of} \ \bq[:,\psi_{\beta}(\cT)]}$};

\path ($(Lsupp.south)+(-2.15*\x,\x*0.3)$) edge[bend right=5, -{Latex[length=2mm,width=1mm]}]  ($(qNW)+(\x*0.25,\x*0.1)$) ;
\path ($(Lsupp.south)+(-1.8*\x,\x*0)$) edge[bend right=5, -{Latex[length=2mm,width=1mm]}]  ($(qNW)+(\x*1.25,\x*0.1)$) ;
\path ($(Lsupp.south)+(1.3*\x,\x*0)$) edge[bend left=5, -{Latex[length=2mm,width=1mm]}]  ($(qNW)+(\x*5.25,\x*0.1)$) ;
\path ($(Lsupp.south)+(1.5*\x,\x*0)$) edge[bend left=5, -{Latex[length=2mm,width=1mm]}]  ($(qNW)+(\x*5.75,\x*0.1)$) ;
\path ($(Lsupp.south)+(2.15*\x,\x*0.3)$) edge[bend left=5, -{Latex[length=2mm,width=1mm]}]  ($(qNW)+(\x*7.25,\x*0.1)$) ;

\coordinate (SLM2) at ($(SLM) + (-0.7*\x,-4*\x)$); %
\node at ($(SLM2)+(-0.5*\x,0*\x)$) {$=$};

\draw (SLM2) -- ++(0.2*\x,1.6*\x);
\draw (SLM2) -- ++(0.2*\x,-1.6*\x);

\coordinate (SRM2) at ($(SLM2)+(17*\x,0*\x)$);
\draw (SRM2) -- ++(-0.2*\x,1.6*\x);
\draw (SRM2) -- ++(-0.2*\x,-1.6*\x);
\node[anchor=west] at ($(SRM2)+(-0.15*\x,-1.5*\x)$) {$\mathsf{row}$};

\coordinate (Gbeta2) at ($(SLM2)+(1*\x,0*\x)$); %

\coordinate (GNW2) at ($(Gbeta2)+(-0.5*\x,0.75*\x)$); %

\draw[draw=none, fill=TUMBlue!20] (GNW2) rectangle ++(1.5*\x,-0.5*\x);
\node at ($(GNW2) + (0.25*\x,-0.25*\x)$) {$1$};
\node at ($(GNW2) + (0.75*\x,-0.25*\x)$) {$1$};
\node at ($(GNW2) + (1.25*\x,-0.25*\x)$) {$1$};
\draw[draw=none, fill=TUMBlue!20] ($(GNW2)+(1.5*\x,-0.5*\x)$) rectangle ++(1.5*\x,-.5*\x);
\node at ($(GNW2) + (1.75*\x,-0.75*\x)$) {$1$};
\node at ($(GNW2) + (2.25*\x,-0.75*\x)$) {$1$};
\node at ($(GNW2) + (2.75*\x,-0.75*\x)$) {$1$};
\draw[draw=none, fill=TUMBlue!20] ($(GNW2)+(3*\x,-1*\x)$) rectangle ++(1.5*\x,-.5*\x);
\node at ($(GNW2) + (3.25*\x,-1.25*\x)$) {$1$};
\node at ($(GNW2) + (3.75*\x,-1.25*\x)$) {$1$};
\node at ($(GNW2) + (4.25*\x,-1.25*\x)$) {$1$};

\draw[draw=none, fill=TUMBlue!20] ($(GNW2)+(4.5*\x,0*\x)$) rectangle ++(3*\x,-1.5*\x);

\draw[dashed, color=TUMGray] ($(GNW2)+(1.5*\x,0*\x)$) -- ++(0*\x,-1.5*\x);
\draw[dashed, color=TUMGray] ($(GNW2)+(3*\x,0*\x)$) -- ++(0*\x,-1.5*\x);
\draw[dashed, color=TUMGray] ($(GNW2)+(4.5*\x,0*\x)$) -- ++(0*\x,-1.5*\x);
\draw[dashed, color=TUMGray] ($(GNW2)+(6*\x,0*\x)$) -- ++(0*\x,-1.5*\x);

\draw[] (GNW2) rectangle ++(7.5*\x,-1.5*\x);

\node at ($(SLM2)+(8.5*\x,0*\x)$) {$\odot$};

\coordinate (qNW2) at ($(SLM2)+(9*\x,1.5*\x)$); %

\foreach \i in {1.5,4}{
\draw[draw=none, fill=TUMBlue!20] ($(qNW2)+(\i*\x,0*\x)$) rectangle ++(0.5*\x,-3*\x);
}

\draw[draw=none, fill=TUMBlue!20] (qNW2) rectangle ++(0.5*\x,-0.5*\x);
\node at ($(qNW2) + (0.25*\x,-0.25*\x)$) {$1$};
\draw[draw=none, fill=TUMBlue!20] ($(qNW2)+(1*\x,-0.5*\x)$) rectangle ++(0.5*\x,-0.5*\x);
\node at ($(qNW2) + (1.25*\x,-0.75*\x)$) {$1$};
\draw[draw=none, fill=TUMBlue!20] ($(qNW2)+(5*\x,-1*\x)$) rectangle ++(0.5*\x,-0.5*\x);
\node at ($(qNW2) + (5.25*\x,-1.25*\x)$) {$1$};
\draw[draw=none, fill=TUMBlue!20] ($(qNW2)+(5.5*\x,-1.5*\x)$) rectangle ++(0.5*\x,-0.5*\x);
\node at ($(qNW2) + (5.75*\x,-1.75*\x)$) {$1$};
\draw[draw=none, fill=TUMBlue!20] ($(qNW2)+(7*\x,-2*\x)$) rectangle ++(0.5*\x,-0.5*\x);
\node at ($(qNW2) + (7.25*\x,-2.25*\x)$) {$1$};

\node[color=TUMGray,rotate=90] at ($(qNW2)+(0.75*\x,-1.5*\x)$) {$\bq[:,\psi_{\beta}(1)]$};
\node[color=TUMGray,rotate=90] at ($(qNW2)+(2.25*\x,-1.5*\x)$) {$\bq[:,\psi_{\beta}(2)]$};
\node[color=TUMGray] at ($(qNW2)+(3.75*\x,-1.5*\x)$) {$\cdots$};

\draw[dashed, color=TUMGray] ($(qNW2)+(1.5*\x,0*\x)$) -- ++(0*\x,-3*\x);
\draw[dashed, color=TUMGray] ($(qNW2)+(3*\x,0*\x)$) -- ++(0*\x,-3*\x);
\draw[dashed, color=TUMGray] ($(qNW2)+(4.5*\x,0*\x)$) -- ++(0*\x,-3*\x);
\draw[dashed, color=TUMGray] ($(qNW2)+(6*\x,0*\x)$) -- ++(0*\x,-3*\x);

\draw (qNW2) rectangle ++(7.5*\x,-3*\x);

\coordinate (L3) at ($(SLM) + (0*\x,-7.5*\x)$); %
\node at ($(L3)+(4*\x,0*\x)$) {$\ni$};

\coordinate (U1) at ($(L3)+(4.5*\x,0.25*\x)$); %

\node at ($(U1)+(3.75*\x,0.5*\x)$) {Star-product of first rows};
\draw[draw=none,fill=TUMBlue!20] (U1) rectangle ++(.5*\x,-.5*\x) node[pos=.5] {$1$};
\draw[] (U1) rectangle ++(7.5*\x,-.5*\x);

\end{tikzpicture}
    \caption[Rank properties of the Khatri-Rao product of matrices.]{Illustration of the proof of Lemma~\ref{lem:rankKhatri} for $k=t=\beta=3$. The \textcolor{TUMBlue}{blue} areas indicate positions that are potentially nonzero, white areas contain only zeros. Columns in the support of $\bq$, i.e., nonzero columns of $\bq$, are indicated by ${\color{TUMBlue}\neq \mathbf{0}}$. The second line corresponds to \eqref{eq:stepIllustrated} when $\bP_1$ and $\bP_2$ are chosen as described in the proof for the set $\cT=\{1,4,5\}$. The third line is the first unit vector, given by the star-product of the first rows of the two matrices.}
    \label{fig:illustrationRankProof}
  \end{figure*}

  Now the first row of the matrix $(\bP_1 \cdot \bG)\otimes \mathbf{1}_\beta$ is a vector that is only (and exactly) supported on the positions $\psi_\beta(\cT)$. Further, by the choice of $\bP_2$, for any $l\in [\beta] \cap \colsupp(\bq[:,\psi_{\beta}(\cT)]$ there exists a row in the matrix $\bP_2 \cdot \bq$ of support $\cS\subset \{l\} \cup \psi_\beta(\{2,3,\ldots,k\})$ and $l\in \cS$ (for an illustration, see Figure~\ref{fig:illustrationRankProof}). Hence, the star-product of these rows, which by definition of the column-wise Khatri-Rao product is a row of $(\bP_1 \cdot (\bG\otimes \mathbf{1}_\beta)) \odot (\bP_2 \cdot \bq)$, is the $l$-th unit vector.

  By the same approach we can show that all the unit vectors $\be_l\in \F^{(k+t-1)\beta}, l \in \colsupp(\bq)$ are contained in the row span\footnote{Observe that the matrices $\bP_1$ and $\bP_2$ are chosen to show that a specific unit vector is contained as a row of the matrix $(\bP_1 \cdot (\bG\otimes \mathbf{1}_\beta)) \odot (\bP_2 \cdot \bq)$, which implies that it also in the span of $(\bG\otimes \mathbf{1}_\beta) \odot \bq$. As we are interested in showing which unit vectors are in the \emph{span}, we do \emph{not} require the matrices $\bP_1$ and $\bP_2$ to be the same for all unit vectors $\be_l\in \F^{(k+t-1)\beta}, l \in \colsupp(\bq)$. Instead, it suffices that for each of these unit vectors there \emph{exists} a choice of $\bP_1$ and $\bP_2$ such that it is a row of the resulting matrix.} of $(\bG\otimes \mathbf{1}_\beta) \odot \bq$ and the lemma statement follows.
\end{proof}

\begin{remark}
We consider the application of Lemma~\ref{lem:rankKhatri} to the two special cases of $t=1$ and $k=1$, \emph{i.e.}, the case of no collusion and the uncoded  (repetition-coded) setting:
\begin{itemize}
    \item $t=1$: In this case, the matrix $\bG$ is simply a full-rank $k\times k$ matrix spanning the trivial $[k,k]$ code, \emph{i.e.}, the full space $\F_q^k$. The matrix $\bP_1$ is given by $\bG^{-1}$ and the $\beta k$ columns of $(\bP_1 \cdot \bG)\otimes \mathbf{1}_\beta$ are $\beta$ repetitions of each unit vector. Hence, the matrix $((\bP_1 \cdot \bG)\otimes \mathbf{1}_\beta) \odot (\bP_2 \cdot \bq)$ is a block diagonal matrix, where each block on the diagonal is a thick column of $\bP_2 \cdot \bq$. Clearly, the columns of different blocks are linearly independent and therefore the rank of this matrix is the sum over the rank of each thick column of $\bq$. If \eqref{eq:rankRestrictionMatrix} holds, this is exactly the cardinality of the column support of $\bq$.
    \item $k=1$: In this case, the matrix $\bG$ spans a $[t,1]$ repetition code, \emph{i.e.}, is a $1\times t$ matrix with non-zero entries. Hence, the matrix $(\bG\otimes \mathbf{1}_\beta) \odot \bq$ is equal to $\bq$ up to non-zero scalar multiples of the columns. As $\bq[:,\psi_\beta(\cT)]= \bq$ for $k=1$, the lemma holds trivially.
\end{itemize}
\end{remark}

With this technical lemma established, we can now link the entropy of the answers of any subset of $k+t-1$ servers to the column support of the query.
\begin{lemma}\label{lem:starProductDimensionBoundQueryMDS}
  Let $\cC$ be an $[n,k]$ MDS code with generator matrix $\bG\in \F_q^{k \times n}$ and $\bY = \bX \cdot \bG \in \F_q^{\alpha m \times n}$, where $\bX$ is chosen uniformly at random from all $\F^{\alpha m \times k}$ matrices. Further, let $\bq \in \F^{\alpha m \times \beta n}$ be a matrix such that for any set $\mathcal{T}\subset [n]$ with $|\cT|=t$, and nonempty set $\cF \subset [m]$ we have
  \begin{align*}
    \rank(\bq[\psi_{\alpha}(\cF),\psi_{\beta}(\cT)]) = |\colsupp(\bq[\psi_{\alpha}(\cF),\psi_{\beta}(\cT)])| \ .
  \end{align*}
  Then for any set $\cN \subset [n]$ with $|\cN|=k+t-1$ it holds that
  \begin{align*}
    H\Big( \sum_{l \in \psi_{\alpha}(\cF)} (\bY[l,\cN]  \otimes \mathbf{1}_\beta&) \star \bq[\psi_{\alpha}(\cF),\psi_{\beta}(\cN)]\Big) \\
     &= |\colsupp(\bq[\psi_{\alpha}(\cF),\psi_{\beta}(\cN)])| \ ,
  \end{align*}
  where $\bq[\psi_{\alpha}(\cF),\psi_{\beta}(\cN)]$ denotes $\bq$ restricted to the $\beta$-thick columns indexed by $\cN$ and $\alpha$-thick rows indexed by $\cF$.
\end{lemma}
\begin{proof}
Let $\bI_{m}$ denote the $m \times m$ identity matrix. We begin with some transformation steps:
\begin{align*}
  &\sum_{l \in \psi_{\alpha}(\cF)} (\bY[l,\cN] \otimes \mathbf{1}_\beta) \star \bq[\psi_{\alpha}(\cF),\psi_{\beta}(\cN)] \\
  &\stackrel{\phantom{(\ref{eq:mixedProduct})}}{=} \mathbf{1}_{|\psi_{\alpha}(\cF)|} \cdot \big(((\bX[\psi_{\alpha}(\cF),:] \cdot \bG_\cN)\otimes \mathbf{1}_\beta) \\
   &\hspace{4cm}\star ( \bI_{|\psi_{\alpha}(\cF)|} \cdot \bq[\psi_{\alpha}(\cF),\psi_{\beta}(\cN)] )\big) \\
  &\stackrel{(\ref{eq:assoKronecker})}{=}  \mathbf{1}_{|\psi_{\alpha}(\cF)|} \cdot \big((\bX[\psi_{\alpha}(\cF),:] \cdot (\bG_\cN\otimes \mathbf{1}_\beta)) \\
  &\hspace{4cm}\star ( \bI_{|\psi_{\alpha}(\cF)|} \cdot \bq[\psi_{\alpha}(\cF),\psi_{\beta}(\cN)])\big)  \\
  &\stackrel{(\ref{eq:mixedProduct})}{=} \mathbf{1}_{|\psi_{\alpha}(\cF)|} \cdot \big((\bX[\psi_{\alpha}(\cF),:] *  \bI_{|\psi_{\alpha}(\cF)|})\\
  &\hspace{3.4cm}\cdot ((\bG_\cN\otimes \mathbf{1}_\beta) \odot \bq[\psi_{\alpha}(\cF),\psi_{\beta}(\cN)])\big)  \\
  &\stackrel{\phantom{(\ref{eq:mixedProduct})}}{=} \big( \mathbf{1}_{|\psi_{\alpha}(\cF)|} \cdot (\bX[\psi_{\alpha}(\cF),:] *  \bI_{|\psi_{\alpha}(\cF)|}) \big) \\
    &\hspace{3.3cm}\cdot \big((\bG_\cN\otimes \mathbf{1}_\beta) \odot \bq[\psi_{\alpha}(\cF),\psi_{\beta}(\cN)]\big)  \ .
\end{align*}
By (\ref{eq:vectorize}) and the definition of $\bX$, the vector
\begin{align*}
\mathbf{1}_{|\psi_{\alpha}(\cF)|} \cdot (\bX[\psi_{\alpha}(\cF),:] *  \bI_{|\psi_{\alpha}(\cF)|})
\end{align*}
is uniformly distributed over $\F^{1\times k |\psi_{\alpha}(\cF)|}$ and it follows that
\begin{align*}
 &H\Big(\sum_{l \in \psi_{\alpha}(\cF)} (\bY[l,\cN] \otimes \mathbf{1}_\beta) \star \bq[l,\psi_{\beta}(\cN)]\Big) \\
  &= H\Big(\big( \mathbf{1}_{|\psi_{\alpha}(\cF)|} \cdot (\bX[\psi_{\alpha}(\cF)] *  \bI_{|\psi_{\alpha}(\cF)|}) \big) \\
  &\hspace{3cm} \cdot \big((\bG_\cN\otimes \mathbf{1}_\beta) \odot \bq[\psi_{\alpha}(\cF),\psi_{\beta}(\cN)]\big)  \Big)\\
   &= \rank\Big((\bG_\cN\otimes \mathbf{1}_\beta) \odot \bq[\psi_{\alpha}(\cF),\psi_{\beta}(\cN)]\Big) \\
 &= |\colsupp(\bq[\psi_{\alpha}(\cF),\psi_{\beta}(\cN)])| \ ,
\end{align*}
where the last equality holds by Lemma~\ref{lem:rankKhatri}.
\end{proof}

One key to the proof of the capacity of \newName{} schemes is that while it is generally not possible to make a statement on the expected rank of a query solely based on the requirement that a PIR scheme is $t$-private, it \emph{is possible} to make such a statement on the expected size of the support of the query.

\begin{lemma}\label{lem:expectationSupport}
  For any PIR scheme, file indices $i,i' \in [m]$, and any $\cF \subset [m]$ it holds that
  \begin{align*}
    &\underset{\mathbf{q}\in \supp(Q^{i})}{\mathbb{E}}|\colsupp(\mathbf{q}[\psi_{\alpha}(\cF),:] )| \\
    &\hspace{3cm}=  \underset{\mathbf{q}\in \supp(Q^{i'})}{\mathbb{E}}|\colsupp(\mathbf{q}[\psi_{\alpha}(\cF),:] )|   \ .
  \end{align*}
\end{lemma}
\begin{proof}
  As the scheme is private, the query $Q_j^i$ to each individual server $j\in[n]$ must be independent of the index $i$, \emph{i.e.}, $\bQ^i[:,\psi_\beta(j)]$ and $\bQ^{i'}[:,\psi_\beta(j)]$ must have the same probability distribution. Trivially, this implies that, the $(|\cF|\alpha\times \beta)$-matrices $\bQ^{i}[\psi_{\alpha}(\cF), \psi_{\beta}(j)] $ and $\bQ^{i'}[\psi_{\alpha}(\cF), \psi_{\beta}(j)]$ also have the same probability distribution and therefore
  \begin{align*}
    &\underset{\bq\in \supp(Q^{i})}{\mathbb{E}}\big(\big|\colsupp(\bq[\psi_\alpha(\cF),:])\cap \psi_\beta(j)\big|\big) \\
    &\hspace{1cm}= \underset{\bq\in \supp(Q^{i'})}{\mathbb{E}}\big(\big|\colsupp(\bq[\psi_\alpha(\cF),:])\cap \psi_\beta(j)\big|\big) \ .
  \end{align*}
  Writing the column support as a disjoint union, we get
  \begin{align*}
|\colsupp(\mathbf{q}[\psi_{\alpha}(\cF),:] )| \! = \! \! \sum_{j\in[n]}\! |\colsupp(\mathbf{q}[\psi_{\alpha}(\cF),:] )\cap \psi_\beta(j)|,
  \end{align*}
  and so by additivity of the expectation we have
 \begin{align*}
   &\underset{\mathbf{q}\in \supp(Q^{i})}{\mathbb{E}}|\colsupp(\mathbf{q}[\psi_{\alpha}(\cF),:] )| \\
   &\hspace{1cm} = \sum_{j\in[n]}  \underset{\mathbf{q}\in \supp(Q^{i})}{\mathbb{E}}|\colsupp(\mathbf{q}[\psi_{\alpha}(\cF),:]) \cap \psi_\beta(j) | \\
   &\hspace{1cm} = \sum_{j\in[n]}  \underset{\mathbf{q}\in \supp(Q^{i'})}{\mathbb{E}}|\colsupp(\mathbf{q}[\psi_{\alpha}(\cF),:]) \cap \psi_\beta(j) | \\
   &\hspace{1cm} = \underset{\mathbf{q}\in \supp(Q^{i'})}{\mathbb{E}}|\colsupp(\mathbf{q}[\psi_{\alpha}(\cF),:] )| \ .
 \end{align*}
\end{proof}

\section{Refined and Lifted PIR Schemes}\label{app:liftedFix}

In this appendix we aim to clarify some of the details of the refinement operation of \cite{oliveira2019one}. Specifically, this operation is based on choosing vectors such that their respective inner product with the stored vectors are ``linearly independent random variables''. Given the application of these rules in \cite[Example~7]{oliveira2019one}, this appears to mean that the corresponding columns in the column-wise Khatri-Rao product of the matrix of storage vectors and the matrix of the query vectors are linearly independent. However, as we discuss in the following, this is not sufficient for the scheme to be private. To allow for better comparison with \cite{oliveira2019one}, we follow their notation in the following.

\subsection{A Counter-Example Violating Privacy} \label{sec:counterExampleLifted}

We consider \cite[Example~7]{oliveira2019one} for the setting $n=4$ and $k=t=2$ with $m=2$ files. There and in the following, file $1$ is assumed to be desired by the user. The storage code is a $[4,2]$ MDS code over $\F_3$ with generator matrix (cf. \cite[Table~VII]{oliveira2019one})
\begin{align*}
  \bG =
  \begin{pmatrix}
    1 & 0 & 1 & 1\\
    0 & 1 & 1 & 2\\
  \end{pmatrix} \ .
\end{align*}
Considering the linear combinations used to obtain $\bx_3^2$ and $\bx_4^2$, it is easy to see that
\begin{align}\label{eq:queryCodeLifted}
[\bx_1^2,\bx_2^2,\bx_3^2,\bx_4^2] = [\bx_1^2, \bx_2^2] \cdot \bG
\end{align}
and we therefore also refer to the code generated by $\bG$ as the \emph{query code}\footnote{In general, the storage and query code do not need to be the same.}. For the desired file $1$, the vectors $\bx^1_j$ are chosen uniformly at random from all query vectors of $\F_3^{\alpha 2 \times 1}$ supported only on file $1$ (cf. \cite[Definition~1]{oliveira2019one}) and such that the $\myspan{\bY_j,\bx_j^1}$ are ``linearly independent random variables''. For the undesired file $2$, the vectors $\bx^2_1$ and $\bx^2_2$ are chosen uniformly at random from all query vectors supported only on file $2$ and such that $\myspan{\bY_1,\bx_1^2}$ and $\myspan{\bY_2,\bx_2^2}$ are ``linearly independent random variables''. The vectors $\bx_3^2$ and $\bx_4^2$ are given by (\ref{eq:queryCodeLifted}).

We set the subpacketization to be $\alpha=2$, \emph{i.e.}, the storage is a length $4$ vector, where the first two positions correspond to file $1$ and the other two positions to file $2$. Now assume the following realizations of $\bx^1$ and $\bx^2$ (the $j^{\mathrm{th}}$ column of $\bx^l$ gives $\bx^l_j$)
\begin{align*}
  \bx^1 =
  \begin{pmatrix}
    1 & 2 & 0 & 0\\
    0 & 0 & 1 & 2\\ \hdashline
    0 & 0 & 0 & 0\\
    0 & 0 & 0 & 0
  \end{pmatrix}
                \quad
  \bx^2 =
  \begin{pmatrix}
    0 & 0 & 0 & 0\\
    0 & 0 & 0 & 0\\\hdashline
    1 & 2 & 0 & 2\\
    0 & 0 & 0 & 0\\
  \end{pmatrix} \ .
\end{align*}
By \cite[Lemma~1]{oliveira2019one} and written in terms of our notation\footnote{Here, the fourth server only receives one query, so the fourth "thick" column is only one column wide.}, the query is then given by
\begin{align*}
    \bq =  \left(\!\! \begin{array}{cc: cc:cc:c}
    1 & 0 & 2 & 0 & 0 & 0 & 0\\
    0 & 0 & 0 & 0 & 1 & 0 & 2\\ \hdashline
    0 & 1 & 0 & 2 & 0 & 0 & 2\\
    0 & 0 & 0 & 0 & 0 & 0 & 0
  \end{array} \!\! \right) \ ,
\end{align*}
where server $j$ receives the $j^\mathrm{th}$ thick column, as indicated by the dashed lines.
We make the following observations:
\begin{itemize}
\item The (positions in the Khatri-Rao product corresonding to) $\myspan{\bY_j,\bx_j^1}, j=1,2,3,4,$ are indeed linearly independent, as any two columns of the storage code are linearly independent.
\item By the same argument, $\myspan{\bY_1,\bx_1^2}$ and $\myspan{\bY_2,\bx_2^2}$ are linearly independent.
\item The third and fourth columns of $\bx^2$, \emph{i.e.}, $\bx^2_3$ and $\bx^2_4$, are as in \eqref{eq:queryCodeLifted}.
\end{itemize}
As $\bx^1$ and $\bx^2$ are chosen uniformly at random such that these properties are fulfilled, this is a query realization with nonzero probability. However, since $\myspan{\bY_3,\bx_3^2} = \myspan{\bY_3,\mathbf{0}} = 0$, the query $\bx_3^2$ is not a valid query if file $2$ is the desired file. Hence, upon receiving the queries $\bx_3^1$ and $\bx_3^2$, server $3$ is able to deduce that file $2$ is not the desired file. Further, observe that here we have $\bx^1[\{1,2\},1]=\bx^2[\{3,4\},1]$ and $\bx^1[\{1,2\},2] = \bx^2[\{3,4\},2]$, so simply excluding this case for the undesired file is not an option, as this would allow servers one and two to deduce that file $1$ is the desired file.

\subsection{High-Level View of the Fixed Scheme}\label{sec:highLevelLiftingFix}

  It is easy to see that the problem described in the previous section is that while $\myspan{\bY_1,\bx_1^2}$ and $\myspan{\bY_2,\bx_2^2}$ give linearly independent random variables, the vectors $\bx_1^2$ and $\bx_2^2$ themselves are not linearly independent. This leads to an $\bx_3^2$ that trivially results in a ``linearly dependent random variable''.

  The additional property required for the scheme of \cite{oliveira2019one} to be private is that the queries received by any $t$-subset of servers leads to ``linear independent random variables'' for \emph{both} files. In the counter-example above this was violated because $\bx_3^2 = \mathbf{0}$. The simplest solution to guaranteeing that this property is fulfilled is controlling the rank of any $t$-subset of the submatrices of $\bx^{1},\bx^2$ corresponding to each file, \emph{i.e.}, in the example given by $\bx^1[\{1,2\},:]$ and $\bx^2[\{3,4\},:]$. In particular, choosing these matrices uniformly random from all matrices generating a given $[n,t]$ MDS codes
  ensures that every subset of $t$ columns is of full rank $t$. In this case, it is easy to see that the inner products $\myspan{\bY_1,\bx_1^l},\myspan{\bY_2,\bx_2^l},\dots,\myspan{\bY_t,\bx_t^l}$ are linearly independent for both $l\in\{1,2\}$. Specifically, consider the subset $\cT\subset [n]$ with $|\cT|=t$. Then, by \cite[Lemma~6]{randriambololona2013upper} the dimension of the space spanned by these inner products is
  \begin{align*}
    \dim(\myspanRow{\bG_\cT} \star \myspanRow{\bx_\cT^l}) = t \ ,
  \end{align*}
  which implies their independence. This choice for $\bx^1,\bx^2$ also guarantees the privacy of the scheme. Since the submatrix corresponding to each file is chosen randomly from all matrices generating the MDS code, the part received by each $t$-subset of servers is a full-rank $t\times t$ matrix uniformly distributed over the set of all full-rank matrices in $\F^{t\times t}$.

  While this ensure privacy, we need to make sure that it preserves retrievability of the desired symbols. Here, the critical property that allows for the increase in rate is the difference in dimension of the star-product between the query for the desired and undesired file. First, consider the desired file and w.l.o.g. assume this to be file $1$. For this file the goal is to ensure that the inner products $\myspan{\bY_1,\bx_1^1},\myspan{\bY_2,\bx_2^1},\dots,\myspan{\bY_n,\bx_n^1}$ are independent or, equivalently, for the space $\myspanRow{\bG_\cT} \star \myspanRow{\bx^1}$ to be of large dimension. As noted in \cite[Proof of Lemma~1]{oliveira2019one}, this is a generic property and easily satisfied over a large enough field.

  The basis for the code of the undesired file needs to be chosen according to the one-shot scheme being refined. An explicit method to choose this code is obtained, \emph{e.g.}, by using the star-product scheme of \cite{Freij-Hollanti2017} as the one-shot scheme.
  For the undesired file, the dimension of the inner products $\myspan{\bY_1,\bx_1^2},\myspan{\bY_2,\bx_2^2},\dots,\myspan{\bY_n,\bx_n^2}$ is supposed to be as \emph{small} as possible, which is guaranteed in the star-product scheme \cite{Freij-Hollanti2017} by using a GRS code with the same code locators as the GRS storage code (see also Section~\ref{sec:strongly}). By choosing the part of $\bx^2$ corresponding to file $2$ such that it generates this code, we obtain
  \begin{align*}
    \dim(\myspanRow{\bG} \star \myspanRow{\bx^2}) = k+t-1 \ .
  \end{align*}
This implies that \emph{all} the inner products $\myspan{\bY_1,\bx_1^2},\myspan{\bY_2,\bx_2^2},\dots,\myspan{\bY_n,\bx_n^2}$ can be obtained from just a subset of $k+t-1$ of these inner products. In turn, this enables the gain of the refinement lemma, as $n-(k+t-1)$ queries can be saved by querying with sums of columns of $\bx^1$ and $\bx^2$ instead of individual columns, as will be discussed in more detail in the following example.

In conclusion, a ``fix'' to the ambiguity in the choice of the matrices $\bx$, which ensures the privacy of this scheme, is given by requiring that the supported columns of any subset of $t$ thick columns of each $\bx^l$ are linearly independent, exactly as required in Definition~\ref{def:newPIRproperty}. Note that our proposed fix allows the scheme to achieve the highest rate possible (for this specific scheme, not in general). Hence, albeit it might be possible to find a different distribution that also results in a private version of the scheme in \cite{oliveira2019one}, there is no advantage to be gained in terms of rate.

\subsection{Example of the Fixed Scheme}\label{sec:exampleLiftingFix}

We now give an updated version of the refinement for $m=2$ files in \cite[Example~7]{oliveira2019one} and address the subsequent lifting operation to $m>2$ files of \cite[Section~V.A]{oliveira2019one}. Recall the system parameters in this example are $n=4$, $k=2$, and $t=2$. For the one-shot scheme being refined we use the star-product scheme of \cite{Freij-Hollanti2017} and set the subpacketization $\alpha=2$.

\subsubsection{Refinement}\label{sec:refinement}

In the following we consider a storage code $\cC$ with the same generator matrix
\begin{align*}
  \bG =
  \begin{pmatrix}
    1 & 0 & 1 & 1\\
    0 & 1 & 1 & 2\\
  \end{pmatrix}
\end{align*}
as in \cite{oliveira2019one}, but over\footnote{The reason for this increase in field size is that the proposed fix requires a $[4,2]$ MDS with the property that the dimension of the star-product with the storage code is the product of their respective dimensions, \emph{i.e.}, equal to $4$. However, despite this being a generic property for MDS codes of larger field size (cf. \cite[Proof of Lemma~1]{oliveira2019one}), the remarkably small field size of $3$ causes none of the combinations of the few $[4,2]$ MDS codes that exist in this field to have this property.} $\F_5$ instead of $\F_3$. Note that another generator matrix of this code is given by
\begin{align*}
  \bG' =
  \begin{pmatrix}
    1 & 1 & 1 & 1\\
    0 & 1 & 3 & 4\\
  \end{pmatrix} \cdot
  \begin{pmatrix}
    1&0&0&0\\
    0&1&0&0\\
    0&0&2&0\\
    0&0&0&3\\
  \end{pmatrix} \ ,
\end{align*}
hence the code $\cC$ is a $[4,2]$ GRS code with code locators $(0,1,3,4)$ and column multipliers $(1,1,2,3)$.

As explained in the previous section, for constructing the query we need to find two codes which result in different dimensions when taking the star-product with the storage code---large dimension for the desired file and small dimension for the undesired file(s).

We begin with the $[n,t]$ MDS code $\cC^{1}$ for the desired file, which we again assume to be file $1$. Recall that the proposed fix requires the property that $\dim(\cC \star \cC^1) = 4$. It is easy to check that this is fulfilled, \emph{e.g.}, by choosing the code $\cC^1$ to be generated by
\begin{align*}
  \bG^1 =
  \begin{pmatrix}
    1 & 0 & 1 & 1\\
    0 & 1 & 1 & 3\\
  \end{pmatrix} \ .
\end{align*}

Next, consider the code $\cC^2$ used for querying the undesired file $2$. This code is chosen according to the one-shot scheme being refined, in our case the star-product scheme of \cite{Freij-Hollanti2017}. For constructing the parts of the query corresponding to the undesired files, this scheme uses an $[n,t]$ GRS code with the same code locators as the storage code and arbitrary column multipliers. For simplicity, we choose $\cC^2 = \cC$ here. It is then easy to check that $\cC\star \cC^2$ is the $3$-dimensional code generated by the matrix
\begin{align*}
\begin{pmatrix}
    1 & 0 & 0 & 4\\
    0 & 1 & 0 & 2\\
    0 & 0 & 1 & 2
  \end{pmatrix} \ .
\end{align*}
Hence, the decoding equation of the scheme is given by
\begin{align}
  \myspan{\bY_4,\bx_4^2} = 4 \myspan{\bY_1,\bx_1^2} + 2 \myspan{\bY_2,\bx_2^2} + 2 \myspan{\bY_3,\bx_3^2} \ . \label{eq:decEqFixedScheme}
\end{align}

Finally, to construct the query, choose the parts of $\bx^1$ and $\bx^2$ to be uniformly random matrices generating\footnote{For the parameters considered here these matrices are simply the generator matrices of the codes.} the codes $\cC^1$ and $\cC^2$, respectively. For example, one valid choice is
\begin{align*}
  \bx^1 =
  \begin{pmatrix}
    1 & 4 & 0 & 3 \\
    3 & 3 & 1 & 4 \\ \hdashline
    0 & 0 & 0 & 0 \\
    0 & 0 & 0 & 0 \\
  \end{pmatrix}
  \quad
  \bx^2 =
  \begin{pmatrix}
    0 & 0 & 0 & 0 \\
    0 & 0 & 0 & 0 \\ \hdashline
    0 & 1 & 1 & 2 \\
    3 & 4 & 2 & 1 \\
  \end{pmatrix} \ .
\end{align*}
The columns of these matrices are used to construct the queries to each server, as in \cite[Table~XVIII]{oliveira2019one} (which is included here in Table~\ref{tab:queryRefined} for the reader's convenience).
\begin{table}
\renewcommand{\arraystretch}{1.3}
  \centering
  
  \caption{Query structure for the example of Section~\ref{sec:refinement}, same as the structure of \cite[Example 7 / Table XVIII]{oliveira2019one}.}
  \begin{tabular}{cccc}
    Server 1 &Server 2&Server 3&Server 4\\ \hline
    $\bx_1^1$&$\bx_2^1$&$\bx_3^1$&$\bx_4^1+\bx_4^2$\\
    $\bx_1^2$&$\bx_2^2$&$\bx_3^2$& \\
  \end{tabular}
  \label{tab:queryRefined}
  
\end{table}
In terms of our notation, the query matrix is therefore given by
\begin{align*}
    \bq =  \left(\!\! \begin{array}{cc: cc:cc:c}
    1 & 0 & 4 & 0 & 0 & 0 & 3 \\
    3 & 0 & 3 & 0 & 1 & 0 & 4 \\ \hdashline
    0 & 0 & 0 & 1 & 0 & 1 & 2 \\
    0 & 3 & 0 & 4 & 0 & 2 & 1 \\
  \end{array} \!\! \right) \ .
\end{align*}
Note that the response generated by the last column, \emph{i.e.}, the response of server $4$, is given by
\begin{align*}
  \myspan{\bY_4,\bq[:,7]} = \myspan{\bY_4, \bx_4^1} + \myspan{\bY_4, \bx_4^2} \ .
\end{align*}
Using the decoding equation of this scheme, as given in \eqref{eq:decEqFixedScheme}, we obtain
\begin{align*}
 &\myspan{\bY_4, \bx_4^1} =  \myspan{\bY_4,\bq[:,7]} - \myspan{\bY_4, \bx_4^2} \\
   &= \myspan{\bY_4,\bq[:,7]} - (4 \myspan{\bY_1,\bx_1^2} + 2 \myspan{\bY_2,\bx_2^2} + 2 \myspan{\bY_3,\bx_3^2}) \\
  &= \myspan{\bY_4,\bq[:,7]}\\
  & \hspace{1cm}- (4 \myspan{\bY_1,\bq[:,2]} + 2 \myspan{\bY_2,\bq[:,4]} + 2 \myspan{\bY_3,\bq[:,6]}) \ .
\end{align*}
As all terms on the right hand side are known after receiving the responses, the user obtains $4$ independent symbols\footnote{The independence of these symbols is guaranteed by the dimension of the star-product between $\cC$ and $\cC^1$ being $\dim(\cC \star \cC^1)=4$.}
\begin{align*}
  \myspan{\bY_1, \bx_1^1}, \myspan{\bY_1, \bx_2^1}, \myspan{\bY_1, \bx_3^1}, \myspan{\bY_1, \bx_4^1}
\end{align*}
and can recover the $\alpha k = 4$ information symbols of the desired file $1$.

\subsubsection{Lifting}\label{sec:lifting}

The second part of the scheme in \cite{oliveira2019one}, which the authors refer to as \emph{lifting}, is concerned with extending this refined scheme for $m=2$ files to any number of files $m>2$.
In the following we discuss the extension of the example discussed above to $m=3$ files, similar to the extension of \cite[Example~7]{oliveira2019one} in \cite[Section~V.A]{oliveira2019one}. The system parameters remain unchanged, except that the subpacketization is increased to $\alpha=8$.

The query structure for this setting is given in \cite[Table~XX]{oliveira2019one} and consists of $27$ single columns, $9$ sums of two columns (two-sums), and $1$ sum of three columns (three-sum) of the matrices $\bx^{l}\in \F^{24\times 16}, l=1,2,3$. The key to the success of the lifting operation is that three of the two-sums, which are required for symmetrization (to guarantee privacy), and the three-sum behave similar to the queries in the two-file example---the parts corresponding to the undesired files are of low dimension while the part corresponding to the desired file is of large dimension. However, this step has a similar problem as the refinement operation (see Appendix~\ref{sec:counterExampleLifted}), as the two-sums used for symmetrization are chosen randomly (non-zero), which could make their linear combination distinguishable from the sums involving the desired file.

To lift the discussed scheme we need a method to choose the matrices $\bx^l, l\in 1,2,3$ such that the scheme is private while preserving retrievability. To this end, we consider the same codes $\cC^1$ and $\cC^2$ as in Appendix~\ref{sec:refinement} and define the permutation matrix
\begin{align*}
    \bP = \begin{small}
\left(
\arraycolsep=4pt
    \begin{array}{cccccccccccccccc}
      1&0&0&0&0&0&0&0&0&0&0&0&0&0&0&0\\
      0&1&0&0&0&0&0&0&0&0&0&0&0&0&0&0\\
      0&0&1&0&0&0&0&0&0&0&0&0&0&0&0&0\\
      0&0&0&1&0&0&0&0&0&0&0&0&0&0&0&0\\ \hdashline
      0&0&0&0&0&0&1&0&0&0&0&0&0&0&0&0\\
      0&0&0&0&0&0&0&1&0&0&0&0&0&0&0&0\\
      0&0&0&0&0&0&0&0&1&0&0&0&0&0&0&0\\
      0&0&0&0&0&1&0&0&0&0&0&0&0&0&0&0\\ \hdashline
      0&0&0&0&0&0&0&0&0&0&0&0&1&0&0&0\\
      0&0&0&0&0&0&0&0&0&0&0&0&0&1&0&0\\
      0&0&0&0&0&0&0&0&0&0&1&0&0&0&0&0\\
      0&0&0&0&0&0&0&0&0&0&0&1&0&0&0&0\\ \hdashline
      0&0&0&0&0&0&0&0&0&0&0&0&0&0&0&1\\
      0&0&0&0&0&0&0&0&0&0&0&0&0&0&1&0\\
      0&0&0&0&0&0&0&0&0&1&0&0&0&0&0&0\\
      0&0&0&0&1&0&0&0&0&0&0&0&0&0&0&0
    \end{array}\right)
    \end{small} \ .
\end{align*}
For the desired file, we choose the part corresponding to the desired file $1$ to be a random basis of the code\footnote{The permutation matrix is required to preserve the labeling of \cite[Table~XX]{oliveira2019one} (see Table~\ref{tab:queryLifted}). There, for each $l=1,2,3$ the columns (in the order of the server they are sent to) $\{\bx^l_1,\bx^l_2,\bx^l_3,\bx^l_4\}$, $\{\bx^l_7,\bx^l_8,\bx^l_9,\bx^l_6\}$, $\{\bx^l_{13},\bx^l_{14},\bx^l_{11},\bx^l_{12}\}$, and $\{\bx^l_{16},\bx^l_{15},\bx^l_{10},\bx^l_{5}\}$ semantically belong together. In contrast, in $\bG^1 \otimes \bI_4$ each subset of $4$ consecutive columns semantically belongs together. The permutation matrix adjusts for this difference.} $\myspan{(\bG^1 \otimes \bI_4)\cdot \bP}$. It is easy to check that
\begin{align*}
  \myspan{(\bG^1 \otimes \bI_4)\cdot \bP} \star \myspan{(\bG \otimes \mathbf{1}_4) \cdot \bP} = \myspan{\big((\bG^1 \odot \bG) \otimes \bI_4\big) \cdot \bP}
\end{align*}
and therefore this star-product is of dimension $4 \cdot \dim(\cC \star \cC^1) = 16$. %

Similarly, for the undesired files $2$ and $3$, the corresponding parts of the matrices $\bx^2$ and $\bx^3$ are each chosen uniformly random from the bases of the code\footnote{Both file $2$ and $3$ can use the same code $\bG^2$ to construct their query.} $\myspan{(\bG^2 \otimes \bI_4)\cdot \bP}$.

\begin{table}
\renewcommand{\arraystretch}{1.3}
  \centering
  \caption{Query structure for the example of Section~\ref{sec:lifting}, same as the structure of \cite[Section V.A / Table XX]{oliveira2019one}.}
  \begin{tabular}{cccc}
    Server 1 &Server 2&Server 3&Server 4\\ \hline
    $\bx_1^1$&$\bx_2^1$&$\bx_3^1$&$\bx_4^1+\bx_4^2$\\
    $\bx_1^2$&$\bx_2^2$&$\bx_3^2$&$\bx_5^1+\bx_4^3$\\
    $\bx_1^3$&$\bx_2^3$&$\bx_3^3$&$\bx_5^2+\bx_5^3$\\ \hdashline
    $\bx_7^1$&$\bx_8^1$&$\bx_9^1+\bx_9^2$&$\bx_6^1$\\
    $\bx_7^2$&$\bx_8^2$&$\bx_{10}^1+\bx_9^3$&$\bx_6^2$\\
    $\bx_7^3$&$\bx_8^3$&$\bx_{10}^2+\bx_{10}^3$&$\bx_6^3$\\ \hdashline
    $\bx_{13}^1$&$\bx_{14}^1+\bx_{14}^2$&$\bx_{11}^1$&$\bx_{12}^1$\\
    $\bx_{13}^2$&$\bx_{15}^1+\bx_{14}^3$&$\bx_{11}^2$&$\bx_{12}^2$\\
    $\bx_{13}^3$&$\bx_{15}^2+\bx_{15}^3$&$\bx_{11}^3$&$\bx_{12}^3$\\ \hdashline
    $\bx_{16}^1+\bx_{16}^2+\bx_{16}^3$&&\\
  \end{tabular}
  \label{tab:queryLifted}
  
\end{table}

Now consider the queries as given in \cite[Table~XX]{oliveira2019one} (which is included here in Table~\ref{tab:queryLifted} for the reader's convenience).
First observe that privacy is preserved as for each $2$-tuple of servers the part of the query corresponding to a given file is uniformly distributed over all $8\times 8$ full-rank matrices\footnote{The symmetry among files is guaranteed by the scheme of \cite{oliveira2019one}. The proposed fix \emph{additionally} guarantees that these matrices are also of full-rank.}. For example, assume the first two servers collude. For each file $l=1,2,3$, the queries received by these servers are made up of the $8$ columns $\bx_j^l, j\in\{1,2,7,8,13,14,15,16\}$. Furthermore, for each file and the given permutation matrix, these $8$ columns contain exactly $2$ columns from each subblock of a random basis of $\myspan{\bG^1 \otimes \bI_4}$ or $\myspan{\bG^2 \otimes \bI_4}$ (here, a subblock is one of the $4$ instances of the matrix $\bG^1$ or $\bG^2$). It is easy to see that this matrix is of full rank $8$ if and only if the corresponding restriction to the $t$ columns within each subblock is of full rank. As the codes $\cC^1$ and $\cC^2$ are MDS, this is always the case.

Note that the strategy described above results in a set of viable (with non-zero probability) queries which is a subset of the original scheme, namely those where the matrices $\bx^l, l=1,2,3$ contain MDS codes in the respective subblocks. Hence, the retrievability of all $16$ symbols of the desired file $1$, which is given by multiple applications of the same process as in Appendix~\ref{sec:refinement}, follows immediately from the arguments in \cite{oliveira2019one}.

\section{A Scheme that does not fulfill Definition~\ref{def:newPIRproperty}} \label{app:counterExample}

In \cite{Sun2018conj}, a linear PIR scheme from $[n=4, k=2]$ MDS-coded storage with $t=2$ colluding servers and $m=2$ files was presented, achieving a PIR rate  $3/5$. This rate exceeds the one in Conjecture  \ref{tconj}, thereby providing a counter-example that disproves it in its full generality. In the following, we briefly introduce this counter-example with a focus on the query construction and show that it does not fulfill Definition~\ref{def:newPIRproperty}.

Each of the two files is assumed to be comprised of $12$ symbols from $\mathbb{F}_p$ for a large prime $p$ and the subpacketization level is set to $\alpha=6$.
Let
\begin{equation*}
\begin{pmatrix}
  \mathbf{V}_1\\\mathbf{V}_2\\ \vdots \\ \mathbf{V}_6
  \end{pmatrix}, ~  \begin{pmatrix}
  \mathbf{U}_0\\\mathbf{U}_1\\ \vdots \\ \mathbf{U}_5
  \end{pmatrix}
\end{equation*}
be two random  full-rank $6\times 6$ matrices over $\mathbb{F}_p$.
Without loss of generality, suppose that the first file is desired.
The queries to servers $1$ and $2$ are given in \eqref{eq:query1} and \eqref{eq:query2}, respectively,
\begin{figure*}[!t]
  \normalsize
\begin{align}
  Q_1^1=\begin{pmatrix}
 \mathcal{L}_{11}(\mathbf{V}_1^T, \mathbf{V}_2^T, \mathbf{V}_3^T) &  \mathcal{L}_{12}(\mathbf{V}_1^T, \mathbf{V}_2^T, \mathbf{V}_3^T)  & & & \mathcal{L}_{13}(\mathbf{V}_1^T, \mathbf{V}_2^T, \mathbf{V}_3^T)\\
& & \mathcal{L}_{11}(\mathbf{U}_0^T, \mathbf{U}_6^T, \mathbf{U}_8^T)& \mathcal{L}_{12}(\mathbf{U}_0^T, \mathbf{U}_6^T, \mathbf{U}_8^T) &  \mathcal{L}_{13}(\mathbf{U}_0^T, \mathbf{U}_6^T, \mathbf{U}_8^T)
  \end{pmatrix}, \label{eq:query1}
\end{align}
\begin{align}
  Q_2^1=\begin{pmatrix}
 \mathcal{L}_{21}(\mathbf{V}_1^T, \mathbf{V}_4^T, \mathbf{V}_5^T) &  \mathcal{L}_{22}(\mathbf{V}_1^T, \mathbf{V}_4^T, \mathbf{V}_5^T)  & & & \mathcal{L}_{23}(\mathbf{V}_1^T, \mathbf{V}_4^T, \mathbf{V}_5^T)\\
& & \mathcal{L}_{21}(\mathbf{U}_0^T, \mathbf{U}_7^T, \mathbf{U}_9^T)& \mathcal{L}_{22}(\mathbf{U}_0^T, \mathbf{U}_7^T, \mathbf{U}_9^T) &  \mathcal{L}_{23}(\mathbf{U}_0^T, \mathbf{U}_7^T, \mathbf{U}_9^T)
  \end{pmatrix}, \label{eq:query2}
\end{align}
\hrulefill
\end{figure*}
where $\mathcal{L}_{ij}(\mathbf{a},\mathbf{b},\mathbf{c})$ denotes some linear combinations of $\mathbf{a}, \mathbf{b}, \mathbf{c}$ (see P1 and P2 in \cite[Pg.~1004]{Sun2018conj} for more details on the requirements on the coefficients of the involved linear combinations), and
\begin{align*}
 & \mathbf{U}_6=\mathbf{U}_1+\mathbf{U}_2, ~\mathbf{U}_7=\mathbf{U}_1+2\mathbf{U}_2,\\
 & \mathbf{U}_8=\mathbf{U}_3+\mathbf{U}_4,~ \mathbf{U}_9=\mathbf{U}_3+2\mathbf{U}_4.
\end{align*}
Note that this definition includes the processing step done at the servers in \cite{Sun2018conj} as part of the query, which is necessary to describe the scheme as a linear scheme as in Definition~\ref{def:linearPIR}.
Then, in our notation for the query, we have for $\cF=\{1\}$ and $\cT=\{1,2\}$
\begin{align*}
   \bq[&\psi_{\alpha}(\cF), \psi_{\beta}(\cT)] = \bq[\psi_{\alpha}(1), \psi_{\beta}(\{1,2\})] = \bq[[6], [12]] \\
    =& {\scriptstyle\big(\mathcal{L}_{11}(\mathbf{V}_1^T, \mathbf{V}_2^T, \mathbf{V}_3^T) \  \mathcal{L}_{12}(\mathbf{V}_1^T, \mathbf{V}_2^T, \mathbf{V}_3^T)  \ \mathbf{0}_{6\times 2}  \ \mathcal{L}_{13}(\mathbf{V}_1^T, \mathbf{V}_2^T, \mathbf{V}_3^T) }\\
  &\hspace{0.7cm} {\scriptstyle \mathcal{L}_{21}(\mathbf{V}_1^T, \mathbf{V}_4^T, \mathbf{V}_5^T) \  \mathcal{L}_{22}(\mathbf{V}_1^T, \mathbf{V}_4^T, \mathbf{V}_5^T)  \ \mathbf{0}_{6\times 2} \ \mathcal{L}_{23}(\mathbf{V}_1^T, \mathbf{V}_4^T, \mathbf{V}_5^T) \big) }\ ,
\end{align*}
where $\mathbf{0}_{6\times 2}$ denotes the $6\times 2$ zero matrix. The matrix $\bq[\psi_{\alpha}(1), \psi_{\beta}(\{1,2\})]$
is a $6\times 10$ matrix with $6$ non-zero columns that are linear combinations of the $5$ vectors $\mathbf{V}_1^T, \mathbf{V}_2^T, \mathbf{V}_3^T, \mathbf{V}_4^T$, and $\mathbf{V}_5^T$.
Therefore, we have
\begin{align*}
  \mbox{rank}(\bq[\psi_{\alpha}(1),& \psi_{\beta}(\{1,2\})])\le 5 \\
  &<6 =|\colsupp(\bq[\psi_{\alpha}(1), \psi_{\beta}(\{1,2\})])| \ ,
\end{align*}
and conclude that the PIR scheme in \cite{Sun2018conj} does not fulfill Definition~\ref{def:newPIRproperty}.

While it might seem excessive to describe a scheme that does \emph{not} fall into the class of \newName{} PIR schemes in this much detail, we would like to point out that this in fact further motivates our definition. The results presented in Section~\ref{sec:capa} show that \emph{the} distinguishing feature of this scheme is in fact the low rank of the queries, when restricting to a subset of thick columns and rows, thereby strongly hinting at what a scheme for general parameters and of a PIR rate that exceeds the one in Conjecture~\ref{tconj} / Theorem~\ref{thm:coded-colluded} must fulfill.

\section{Notation} \label{app:notation}
The notation used in this work is summarized in Table~\ref{tab:notation}.
\begin{table*}[htbp]
\caption{Notation used in this work.}\label{tab:notation}
\begin{center}
\renewcommand{\arraystretch}{2}
\begin{tabular}{|c|l|c|l|}
\hline
$(n, k, d)$ (resp. $[n, k, d]$) & Code parameters of a (resp. linear) code & $m$	& Number of files \\ \hline  $(n, k)$ (resp. $[n, k]$) & Code parameters of an  (resp. a linear) MDS  code & $t$ & Number of colluding servers \\ \hline
$R_m$ (resp. $C_m$)&Rate (resp. Capacity)  of a PIR scheme with $m$ files &$b$ & Number of Byzantine servers \\\hline $R$&$\lim_{m\rightarrow\infty}R_m $ & $r$ &  Number of nonreponsive servers\\\hline
$[b]$ &Set of integers $\{i, 1\le i\le b\}$  &$n$ & Number of servers/code length \\\hline
$k$ & Code dimension of an MDS code & $d$ & Minimum distance of a code\\\hline
$\supp(W)$ & The set of realizations of $W$ with nonzero probability&$W_{j,...,l}$&$\{W_j,W_{j+1},\ldots,W_{l}\}$\\\hline
$\bW = [\bW_1^\top, \bW_2^\top,\ldots ]^\top$&A matrix with the $j^{\mathrm{th}}$ block of rows corresponding to $W_j$&  $W_{\mathcal{T}}$& $\{W_j : j \in \mathcal{T}\}$\\\hline
$\bW = [\bW_1, \bW_2,\ldots ]$ &A matrix with the $j^{\mathrm{th}}$ block of columns corresponding to $W_j$& $\mathbb{F}_q$ ($\mathbb{F}$) & Finite field of $q$ elements \\\hline
$\bW[\cI,:]$&Submatrix of $\bW$ restricted to the rows indexed  by $\cI$&$\psi_{\beta}(\cI)$&$\bigcup\limits_{i \in \cI} \{(i-1)\beta +1,\ldots, i\beta\}$\\\hline
$\bW[:,\cI]$&Submatrix of $\bW$ restricted to the columns indexed  by $\cI$&$\bG$& Generator matrix of a storage code\\\hline
$\bX$& Stands for $[(\bX^1)^\top, (\bX^2)^\top, \ldots, (\bX^m)^\top]^\top$ & $\bX^i\in \mathbb{F}^{\alpha\times k}$& The $i^{\rm{th}}$ file \\\hline
 $\bY^l=\bX^l\cdot\bG$& Encoded version of the $l^{\rm{th}}$ file &  $\bY=\bX\cdot\bG$&  Encoded version of all the files \\\hline $Y_l$& The $l$-th column of $Y$ & $\alpha$& Number of stripes of each file\\\hline
$Y_\mathcal{I}$ & The restriction of $Y$ to the storage of servers indexed by $\mathcal{I}$ &
$H(\cdot)$& Entropy function\\\hline
$\cS$&   Randomness shared by the servers & $\rho_\spir$ & Secrecy rate of SPIR  \\\hline  $S$& Stands for $(S_1,\dots ,S_n)$ &$\cQ$ & Set of all possible queries \\\hline
$S_j\in \cS$& Shared randomness that used by the $j^{\rm{th}}$ server &$\odot$ & The column-wise Khatri-Rao product \\\hline
$\myspanRow{\bW}$ or $\myspanRow{\cC}$ & Row span of the matrix $\bW$ or vectors in the set $\cC$ &$\left\langle~ , ~\right\rangle$& Inner product operation \\\hline
$Q^i=\left(Q_1^i, \ldots, Q_n^i\right)$ & Query when  the $i^{\rm{th}}$ file is requested & $I(~;~)$ & Mutual information\\\hline
$Q_j^i$ & Query sent to the $j^{\rm{th}}$ server when  the $i^{\rm{th}}$ file is requested & $\star$ & Star product \\\hline
$A_j^i$ & Response  from  the $j^{\rm{th}}$ server when  the $i^{\rm{th}}$ file is requested  &$A^i$ & Stands for $\left(A_1^i,\ldots,A_n^i\right)$
\\\hline
$\colsupp(\bW)$ &The set of indices of nonzero columns of $\bW$&$\otimes$ & The Kronecker product
 \\\hline
\end{tabular}
\end{center}
\end{table*}

\end{document}